\documentclass[a4paper]{article}%
\usepackage{amsmath}
\usepackage{amsfonts}
\usepackage{amssymb}
\usepackage{graphicx}
\usepackage[caption=false]{subfig}%
\providecommand{\U}[1]{\protect\rule{.1in}{.1in}}
\newtheorem{theorem}{Theorem}

\newtheorem{corollary}[theorem]{Corollary}

\newtheorem{definition}[theorem]{Definition}

\newtheorem{remark}[theorem]{Remark}

\newenvironment{proof}[1][Proof]{\noindent \textbf{#1.} }{\  \rule{0.5em}{0.5em}}
\begin{document}

\title{\textbf{Robust Wald-type tests in GLM with random design based on
minimum density power divergence estimators}}
\author{Basu, A$^1$; Ghosh, A$^1$; Mandal, A$^2$; Martin, N$^3$ and Pardo, L$^3$\\
$^1$Indian Statistical Institute, Kolkata 700108, India\\
$^2$Wayne State University, Detroit, MI 48202, USA\\
$^3$Complutense University of Madrid, 28040 Madrid, Spain}
\date{\today}
\maketitle

\begin{abstract}
We consider the problem of robust inference under the  generalized
linear model (GLM) with stochastic covariates. We derive the properties of the
minimum density power divergence estimator of the parameters in GLM with
random design and use this estimator to propose robust Wald-type tests for
testing any general composite null hypothesis about the GLM. The asymptotic and
robustness properties of the proposed tests are also examined for the GLM with
random design. Application of the proposed robust inference procedures to the
popular Poisson regression model for analyzing count data is discussed
in detail both theoretically and numerically through simulation studies and real data examples. 
\end{abstract}

\bigskip\bigskip



\section{Introduction\label{sec1}}

Parametric statistical modelling is an important tool in statistical analysis of real data. 
Whenever the parametric assumption is satisfied, the parametric method will be much more efficient than the corresponding non-parametric methods. 
However, classical parametric methods, including those based on the maximum likelihood principle, 
can be very significantly influenced by the presence of outlying observations in the data, even in a very small proportion. 
The data analyst would, therefore, like to construct and use such procedures which exhibit a high degree of robustness 
(in the sense of outlier stability) with little loss in asymptotic efficiency. 
In the current age of big data, the outlier problem is as relevant as ever. 
In this paper we will deal with the robustness issue in case of generalized linear models where the covariates are stochastic (rather than fixed).

Regression analysis is a basic statistical data analysis technique across different disciplines of applied sciences,
which helps us to model a response variable in terms of several associated covariates.
One major application of regression is also in predicting future observations from the values of the model covariates 
as well as in investigating if a covariate has a statistically significant role in explaining the variability in the response.  
The standard linear regression model is the most common one applicable to a continuous response
having a linear relationship with each covariate. 
We consider a much wider class of regression models, namely 
generalized linear models (GLMs),  first introduced by Nelder and Wedderburn
(1972) and later expanded by McCullagh and Nelder (1989);
they  represent a method of extending standard linear regression to incorporate a variety of responses
including distributions of counts, binary or positive values as well as several types of possible relationship
between the response and covariates (under suitable restrictions). 
Here, the observations
$\left(  y_{i},\boldsymbol{x}_{i}^{T}\right)  ^{T}\in\mathbb{R}^{k+1},\text{
}1\leq i\leq n,$
are assumed to be independent and identically distributed (IID) realizations 
of the random variables $(Y,\boldsymbol{X}^{T})^{T}$
in such a way that the conditional distribution of $Y$ given $\boldsymbol{X=x}$ belongs
to the general exponential family of distributions having density function,
with respect to a convenient $\sigma$-finite measure, given by
\begin{equation}
f_{Y|\boldsymbol{X=x}}\left(  y,\theta,\phi\right)  =\exp\left\{
\frac{y\theta-b(\theta)}{a(\phi)}+c\left(  y,\phi\right)  \right\}  ,\ \ \ \ y \in \mathbb{R}, \theta \in \mathbb{R}^p, \phi >0,
\label{1}%
\end{equation}
where the canonical parameter $\theta$ is an unknown measure of location
depending on the predictor $\boldsymbol{x}$ and $\phi$ is a known or unknown
nuisance scale or dispersion parameter typically required to produce standard
errors following Gaussian, gamma or inverse Gaussian distributions. The
functions $a(\phi)$, $b(\theta)$\ and $c\left(  y,\phi\right)  $\ are known.
In particular, $a(\phi)$ is set to $1$ for binomial, Poisson, and negative
binomial distributions (known $\phi$), and it does not enter into the
calculations for standard errors. The mean of the conditional distribution of
$Y$ given $\boldsymbol{X=x}$, 
namely  $\mu_{Y|\boldsymbol{X=x}}(\theta)=E\left[Y|\boldsymbol{X=x}\right] =b^{\prime}(\theta)$, 
is dependent only on $\theta$ 
and is assumed, according to GLMs, to be modeled linearly with respect to
$\boldsymbol{x}$\ through a known link function, $g$, i.e.,%
$g(\mu_{Y|\boldsymbol{X=x}}(\theta))=\boldsymbol{x}^{T}\boldsymbol{\beta},$
where $g$ is a monotone and differentiable function and $\boldsymbol{\beta\in
}\mathbb{R}^{k}$ is an unknown parameter. In this setting, since
$\theta=\theta\left(  \boldsymbol{x}^{T}\boldsymbol{\beta}\right)  $, we shall
also denote the density in  (\ref{1}) by $f\left(  y,\boldsymbol{x}^{T}\boldsymbol{\beta}%
,\phi\right)$.
The statistical problem is then to first estimate the regression coefficients $\boldsymbol{\beta}$ 
and the variance parameter $\phi$ (if unknown) through appropriate estimation methods like maximum likelihood estimation 
and use these estimates for subsequent hypotheses testing and prediction for the underlying research applications.

To clarify the objective of the present paper, we note that the $n\times k$ matrix 
$\mathbb{X} = [\boldsymbol{x}_1, \cdots, \boldsymbol{x}_n]^T$ is referred to as the design matrix 
in the context of regression. As per the above formulations, all rows of this design matrix are 
IID copies of the $p$-dimensional (covariate) random variable $\boldsymbol{X}$. 
Such situations are referred to as the \textit{random design models} which we focus on the present paper. 
Another alternative option, mostly used for planned design of experiments, is the fixed design models
where each row of the design matrix is assumed to be non-stochastic and pre-fixed.
One can verify that for most common applications, if we assume fixed design set-up
while the values of each row actually came from some underlying distributions,
the estimators of $\boldsymbol{\beta}$ would be identical to the random design case 
but their asymptotic properties (including the variance and hence the standard errors) 
may be quite different depending  on the stochastic structure of the true random design matrix
which, in turn, affects the hypotheses testing results and any predictive confidence intervals.
This can be illustrated through a simple example of maximum likelihood estimator (MLE) of $\boldsymbol{\beta}$ 
under the simple linear regression model, a special case of GLM with $f\equiv N(\theta, \phi)$ and the identity link ($g$).
In this particular example, the MLE is $[\mathbb{X}^T\mathbb{X}]^{-1}\mathbb{X}^T\boldsymbol{y}$
with $\boldsymbol{y}=(y_1, \ldots, y_n)^T$ for both of fixed and random designs,
but the asymptotic variance are given, respectively, by $\Sigma_{\rm fix} = n^{-1}\phi[\mathbb{X}^T\mathbb{X}]^{-1}$ 
for fixed design and $\Sigma_{\rm ran} = n^{-1}\phi E[\mathbb{X}^T\mathbb{X}]^{-1}$ for the random design.
Now, suppose the random design is true having  $p=2$, $\boldsymbol{X}=(1, X)^T$ with $E(X)=0$ and $E(X^2)=\sigma_X^2$;
then ideally the asymptotic variance of MLE of $\boldsymbol{\beta}$ should be 
$\Sigma_{\rm ran} =\phi \begin{bmatrix}
\begin{array}{cc}
1 & 0\\
0 & \sigma_X^{-2}
\end{array}
\end{bmatrix}$, 
a constant independent of the observed data. 
However, if one wrongly assume that the design matrix is fixed based on the observed data,
the corresponding asymptotic variance matrix will turn out to be  
$\Sigma_{\rm fix} =\phi \begin{bmatrix}
\begin{array}{cc}
1 & n^{-1}\sum_{i=1}^nx_i \\
n^{-1}\sum_{i=1}^nx_i  & n^{-1}\sum_{i=1}^nx_i^2 
\end{array}
\end{bmatrix}^{-1}$, 
which can clearly be substantially different from the true $\Sigma_{\rm ran}$ based on the observed data 
for finite sample sizes (e.g., $X$ may be heavy tailed), 
and consequently all the inferential results (testing and confidence intervals) could be adversely affected. 
This motivated the study of GLMs having random design matrices separately from the fixed design cases.

However, the usual inference procedures based on the maximum likelihood and the maximum quasi-likelihood
estimators are extremely non-robust against the data contaminations or model misspecification
under both the fixed or random design set up; these have been studied extensively in the literature for different GLMs
and their non-robustness have been demonstrated by several authors
(Hampel et al. 1986; Stefanski et al., 1986; K\"{u}nsch et al., 1989; Morgenthaler, 1992, and many others). 
Modern complex datasets are prone to having outlying observations either due to 
some confounded effects or error in any stage of data processing which, in turn,
yields incorrect statistical results and research insights if a non-robust method is used to analyze them. 
Consequently,  robust procedures for GLMs have been considered to robustify the MLE. 
Stefanski et al. (1986) studied optimally bounded score functions for the GLM. 
They generalized the results obtained by Krasker and Welsch (1982) for classical linear models. 
The robust estimator of Stefanski et al. (1986) is, however, difficult to compute. 
K\"{u}nsch et al. (1989) introduced another estimator, called the conditionally unbiased bounded-influence estimator. 
The development of robust models for the GLM continued with the work of Morgenthaler (1992). 
More recently, Cantoni and Ronchetti (2001) proposed a robust approach based on robust quasi-deviance
functions for estimation and variable selection. 
Another class of estimators are the M-estimators proposed by Bianco and Yohai (1996)  
and further studied by Croux and Haesbroeck (2003) for logistic regression, a special case of GLMs. 
Bianco et al. (2013) proposed general M-estimators for GLM for data sets with missing values in the responses.
Valdora and Yohai (2014) proposed a family of robust estimators for GLM based
on M-estimators after applying a variance stabilizing transformation to the response. 
More recent works on robust inference in GLMs also include Aeberhard et al. (2014) and Marazzi et al. (2019).
Along this line of research,  Ghosh and Basu (2016) presented a robust estimator
assuming a fixed design, based on the density power divergence approach.
In this paper, we will first extend it to the random design GLMs 
and subsequently discuss its properties in developing robust hypotheses testing procedures.
Throughout this paper, our focus will be on \textit{robustness against data contamination} (e.g., outliers) 
among the sample observations and discuss the properties of the proposed estimators and tests 
in respect of safeguarding against such data contamination.

To define our estimator for the random design GLMs as discussed above, 
we note that the observations  $(y_{1},\boldsymbol{x}_{1}^{T})^{T},...,(y_{n},\boldsymbol{x}_{n}^{T})^{T}$ 
indeed form a random sample from $(Y,\boldsymbol{X}^{T})^{T}$
and  
the density function of $y_{i}|\boldsymbol{X=x}_{i}$ is denoted as $f(y_{i},\boldsymbol{x}_{i}^{T}\boldsymbol{\beta},\phi)$. 
For the cases of non-random design with fixed $\boldsymbol{x}_i$,  Ghosh and Basu (2016) considered a particular class of 
$M$-estimators depending  on a tuning parameter $\alpha>0$, which solved the estimating equation %
\[%
{\displaystyle\sum\limits_{i=1}^{n}}
\boldsymbol{\Psi}_{\alpha}(y_{i},\boldsymbol{x}_{i}^{T}\boldsymbol{\beta}%
,\phi)=\boldsymbol{0},
\]
where
\begin{equation}
\boldsymbol{\Psi}_{\alpha}\left(  y_{i},\theta_{i},\phi,\boldsymbol{x}_{i}%
^{T}\boldsymbol{\beta},\alpha\right)  =%
{\displaystyle\int}
\boldsymbol{u}(y,\boldsymbol{x}_{i}^{T}\boldsymbol{\beta},\phi)f^{1+\alpha
}(y,\boldsymbol{x}_{i}^{T}\boldsymbol{\beta},\phi)dy-\boldsymbol{u}%
(y_{i},\boldsymbol{x}_{i}^{T}\boldsymbol{\beta},\phi)f^{\alpha}(y_{i}%
,\boldsymbol{x}_{i}^{T}\boldsymbol{\beta},\phi), \label{2}%
\end{equation}
with $
\boldsymbol{u}(y_{i},\boldsymbol{x}_{i}^{T}\boldsymbol{\beta},\phi
)=\frac{\partial\log f(y_{i},\boldsymbol{x}_{i}^{T}\boldsymbol{\beta},\phi
)}{\partial\boldsymbol{\eta}}
$
and
$
\boldsymbol{\eta}=(\boldsymbol{\beta}^{T},\phi)^{T},
$
if $\phi$ is unknown, and $\boldsymbol{\eta}=\boldsymbol{\beta}$, otherwise.
In Ghosh and Basu (2016) it was established that
\[
\boldsymbol{u}(y_{i},\boldsymbol{x}_{i}^{T}\boldsymbol{\beta},\phi)=%
\begin{pmatrix}
\frac{\partial\log f(y_{i},\boldsymbol{x}_{i}^{T}\boldsymbol{\beta},\phi
)}{\partial\boldsymbol{\beta}}\\
\frac{\partial\log f(y_{i},\boldsymbol{x}_{i}^{T}\boldsymbol{\beta},\phi
)}{\partial\phi}%
\end{pmatrix}
=%
\begin{pmatrix}
K_{1}(y_{i},\boldsymbol{x}_{i}^{T}\boldsymbol{\beta},\phi)\boldsymbol{x}_{i}\\
K_{2}(y_{i},\boldsymbol{x}_{i}^{T}\boldsymbol{\beta},\phi)
\end{pmatrix}
,
\]
for unknown $\phi$, where
\begin{align*}
K_{1}(y_{i},\boldsymbol{x}_{i}^{T}\boldsymbol{\beta},\phi)  &  =\frac
{y_{i}-\mu(\theta_{i})}{\sigma^{2}(\theta_{i})g^{\prime}\left(  \mu(\theta
_{i})\right)  },
~~~ \sigma^{2}(\theta_{i})    =\mathrm{Var}\left[  Y_{i}|\boldsymbol{X=x}%
_{i}\right]  =a(\phi)b^{\prime\prime}(\theta_{i}),\\
K_{2}(y_{i},\boldsymbol{x}_{i}^{T}\boldsymbol{\beta},\phi)  &  =-\frac
{y_{i}\theta_{i}-b\left(  \theta_{i}\right)  }{a^{2}(\phi)}a^{\prime}%
(\phi)+\frac{\partial c\left(  y_{i},\phi\right)  }{\partial\phi}.
\end{align*}
Therefore, defining
\begin{equation}
\gamma_{j,\alpha}(\boldsymbol{x}_{i})=\int K_{j}(y,\boldsymbol{x}_{i}%
^{T}\boldsymbol{\beta},\phi)f^{1+\alpha}(y,\boldsymbol{x}_{i}^{T}%
\boldsymbol{\beta},\phi)dy,\quad\text{for }j=1,2, \label{gamma}%
\end{equation}
we get
\[
\boldsymbol{\Psi}_{\alpha}(y_{i},\boldsymbol{x}_{i}^{T}\boldsymbol{\beta}%
,\phi)=%
\begin{pmatrix}
\left(  \gamma_{1,\alpha}(\boldsymbol{x}_{i})-K_{1}(y_{i},\boldsymbol{x}%
_{i}^{T}\boldsymbol{\beta},\phi)f^{\alpha}(y_{i},\boldsymbol{x}_{i}%
^{T}\boldsymbol{\beta},\phi)\right)  \boldsymbol{x}_{i}\\
\gamma_{2,\alpha}(\boldsymbol{x}_{i})-K_{2}(y_{i},\boldsymbol{x}_{i}%
^{T}\boldsymbol{\beta},\phi)f^{\alpha}(y_{i},\boldsymbol{x}_{i}^{T}%
\boldsymbol{\beta},\phi)
\end{pmatrix}
,
\]
and the estimating equations are given by
\begin{align}%
{\displaystyle\sum\limits_{i=1}^{n}}
\left(  \gamma_{1,\alpha}(\boldsymbol{x}_{i})-K_{1}(y_{i},\boldsymbol{x}%
_{i}^{T}\boldsymbol{\beta},\phi)f^{\alpha}(y_{i},\boldsymbol{x}_{i}%
^{T}\boldsymbol{\beta},\phi)\right)  \boldsymbol{x}_{i}  &  =\boldsymbol{0}%
,\label{3}\\%
{\displaystyle\sum\limits_{i=1}^{n}}
\left(  \gamma_{2,\alpha}(\boldsymbol{x}_{i})-K_{2}(y_{i},\boldsymbol{x}%
_{i}^{T}\boldsymbol{\beta},\phi)f^{\alpha}(y_{i},\boldsymbol{x}_{i}%
^{T}\boldsymbol{\beta},\phi)\right)   &  =0. \label{4}%
\end{align}
Notice that for known $\phi$, the unique estimating equation is (\ref{3}). It
is clear that
\[
\mathrm{E}\left[  \boldsymbol{\Psi}_{\alpha}(Y,\boldsymbol{X}^{T}%
\boldsymbol{\beta},\phi)|\boldsymbol{X}=\boldsymbol{x}\right]  =\boldsymbol{0}%
,
\]
when the conditional distribution of $Y$ given the covariates belongs to the assumed GLM family
and hence the estimators considered in Ghosh and Basu (2016) are conditionally Fisher-consistent
at the model for random design as well. 
In addition, since%
\begin{align}
\mathrm{E}\left[  \boldsymbol{\Psi}_{\alpha}(Y,\boldsymbol{X}^{T}%
\boldsymbol{\beta},\phi)\right]  =\boldsymbol{0}, \label{EQ:func_est}%
\end{align}
these estimators are also unconditionally  Fisher consistent under random design GLMs as well.
Let us denote $\widehat{\boldsymbol{\eta}}_{\alpha}$ as the estimator of $\boldsymbol{\eta}$, 
obtained by solving  equations (\ref{3}) and (\ref{4}),
which we refer to as the minimum density power divergence estimator (MDPDE) of $\boldsymbol{\eta}$. 
Under suitable differentiability properties of the functions $a(\cdot)$, $b(\cdot)$, $c(\cdot)$ and $g(\cdot)$, 
the equations (4) and (5) are indeed the estimating equations for obtaining the MDPDEs  of the parameter $\eta$; 
see Basu et al. (1998), Ghosh and Basu (2013) and Ghosh and Basu (2016) for a general description of the density power divergence as well as the formulation of the divergence in the generalized linear models scenario.  
Ghosh and Basu (2016) derived the asymptotic distribution of  $\widehat{\boldsymbol{\eta}}_{\alpha}$ assuming that
$\boldsymbol{X}_{i}$, $i=1,...,n$, are non-random (fixed design). 

The primary purpose of this paper is to present the
asymptotic distribution as well as the robustness properties of the minimum density power divergence estimator 
$\widehat{\boldsymbol{\eta}}_{\alpha}$ when $\boldsymbol{X}_{i}$, $i=1,...,n$, are generated by  a random design. 
These are seen to be quite different from those developed under the fixed-design set-up in Ghosh and Basu (2016)
and may be hampered in the same way as illustrated earlier for the MLEs if the design matrix is wrongly assumed to be fixed.
Subsequently,  based on the estimator $\widehat{\boldsymbol{\eta}}_{\alpha}$, 
a family of robust Wald-type tests is introduced.
The properties of the test statistics depend directly on the newly derived properties of the estimator;
we study the asymptotic and robustness properties 
along with appropriate numerical illustrations.

The structure of the paper is as follows. In
Section \ref{sec2} we present the asymptotic distribution of the MDPDE of
$\boldsymbol{\eta}$ for the random design case. Section \ref{sec3} 
introduces Wald-type tests for testing general linear hypothesis on parameters
under study and establishes their asymptotic distribution. The robustness
properties of the Wald-type tests are
studied in Section \ref{sec4}. The Poisson regression model under the random
design is studied in Section \ref{sec5},  and finally, Section \ref{sec6} presents a detailed simulation study illustrating the benefits  of our proposal.

\section{Properties of the MDPDEs under Random Design\label{sec2}}

Together with the notation of Section 1, 
let us assume that $\boldsymbol{X}$ represents the vector of $k$ (random) explanatory variables
and the marginal distribution of $\boldsymbol{X}$ is denoted by $G(\boldsymbol{x})$. 
In the following we first consider the asymptotic properties of the MDPDE and
thereafter, study the  corresponding robustness properties.

\subsection{Asymptotic Properties \label{sec2.1}}

In order to derive the asymptotic distribution of $\widehat{\boldsymbol{\eta}%
}_{\alpha}$, we are going to follow the same scheme as  given in Theorem 10.7 of
Maronna et al. (2006)  for M-estimators. Through this, the asymptotic distribution of
$\widehat{\boldsymbol{\eta}}_{\alpha}$ is  given by%
\[
\sqrt{n}(\widehat{\boldsymbol{\eta}}_{\alpha}-\boldsymbol{\eta}_{0}%
)\underset{n\rightarrow\infty}{\longrightarrow}\mathcal{N}(\boldsymbol{0}%
,\boldsymbol{\Sigma}_{\alpha}(\boldsymbol{\eta}_{0})),
\]
where
$
\boldsymbol{\Sigma}_{\alpha}(\boldsymbol{\eta}_{0})=\boldsymbol{J}_{\alpha
}^{-1}(\boldsymbol{\eta}_{0})\boldsymbol{K}_{\alpha}(\boldsymbol{\eta}%
_{0})\boldsymbol{J}_{\alpha}^{-1}(\boldsymbol{\eta}_{0}),
$
with
\begin{align*}
\boldsymbol{K}_{\alpha}(\boldsymbol{\eta})  &  =\mathrm{E}\left[
\boldsymbol{\Psi}_{\alpha}(Y,\boldsymbol{X}^{T}\boldsymbol{\beta}%
,\phi)\boldsymbol{\Psi}_{\alpha}^{T}(Y,\boldsymbol{X}^{T}\boldsymbol{\beta
},\phi)\right]  =%
{\displaystyle\int_{\mathcal{X}}}
\mathrm{E}\left[  \boldsymbol{\Psi}_{\alpha}(Y,\boldsymbol{x}^{T}%
\boldsymbol{\beta},\phi)\boldsymbol{\Psi}_{\alpha}^{T}(Y,\boldsymbol{x}%
^{T}\boldsymbol{\beta},\phi)\right]  dG(\boldsymbol{x}),\\
\boldsymbol{J}_{\alpha}(\boldsymbol{\eta})  &  =\mathrm{E}\left[
\frac{\partial\boldsymbol{\Psi}_{\alpha}(Y,\boldsymbol{X}^{T}\boldsymbol{\beta
},\phi)}{\partial\boldsymbol{\eta}^{T}}\right]  =%
{\displaystyle\int_{\mathcal{X}}}
\mathrm{E}\left[  \frac{\partial\boldsymbol{\Psi}_{\alpha}(Y,\boldsymbol{X}%
^{T}\boldsymbol{\beta},\phi)}{\partial\boldsymbol{\eta}^{T}}\right]
dG(\boldsymbol{x}).
\end{align*}
Here, $\mathcal{X}$ is the sample space of $\boldsymbol{X}$. 
After some algebra,  the  expressions  turn out to be
\[
\boldsymbol{K}_{\alpha}(\boldsymbol{\eta})=%
\begin{pmatrix}
{\int_{\mathcal{X}}}\left(  \gamma_{11,2\alpha}(\boldsymbol{x})-\gamma
_{1,\alpha}^{2}(\boldsymbol{x})\right)  \boldsymbol{x}\boldsymbol{x}%
^{T}dG(\boldsymbol{x}) & {\int_{\mathcal{X}}}\left(  \gamma_{12,2\alpha
}(\boldsymbol{x})-\gamma_{1,\alpha}(\boldsymbol{x})\gamma_{2,\alpha
}(\boldsymbol{x})\right)  \boldsymbol{x}dG(\boldsymbol{x})\\
{\int_{\mathcal{X}}}\left(  \gamma_{12,2\alpha}(\boldsymbol{x})-\gamma
_{1,\alpha}(\boldsymbol{x})\gamma_{2,\alpha}(\boldsymbol{x})\right)
\boldsymbol{x}^{T}dG(\boldsymbol{x}) & {\int_{\mathcal{X}}}\left(
\gamma_{22,2\alpha}(\boldsymbol{x})-\gamma_{2,\alpha}^{2}(\boldsymbol{x}%
)\right)  dG(\boldsymbol{x})
\end{pmatrix}
,
\]
and
\[
\boldsymbol{J}_{\alpha}(\boldsymbol{\eta})=%
\begin{pmatrix}
{\int_{\mathcal{X}}}\gamma_{11,\alpha}(\boldsymbol{x})\boldsymbol{x}%
\boldsymbol{x}^{T}dG(\boldsymbol{x}) & {\int_{\mathcal{X}}}\gamma_{12,\alpha
}(\boldsymbol{x})\boldsymbol{x}dG(\boldsymbol{x})\\
{\int_{\mathcal{X}}}\gamma_{12,\alpha}(\boldsymbol{x})\boldsymbol{x}^T%
dG(\boldsymbol{x}) & {\int_{\mathcal{X}}}\gamma_{22,\alpha}(\boldsymbol{x}%
)dG(\boldsymbol{x})
\end{pmatrix}
,
\]
where $\gamma_{j,\alpha}(\boldsymbol{x})$, $j=1,2$, is given by (\ref{gamma})
and
\[
\gamma_{jh,\alpha}(\boldsymbol{x})=\int K_{j}\left(  y,\boldsymbol{x}%
^{T}\boldsymbol{\beta},\phi\right)  K_{h}\left(  y,\boldsymbol{x}%
^{T}\boldsymbol{\beta},\phi\right)  f^{1+\alpha}\left(  y,\boldsymbol{x}%
^{T}\boldsymbol{\beta},\phi\right)  dy\text{, for }j,h=1,2.
\]
Notice that for the case where $\phi$ is known, we get $\boldsymbol{K}_{\alpha
}(\boldsymbol{\eta})={\int_{\mathcal{X}}}\left(  \gamma_{11,2\alpha
}(\boldsymbol{x})-\gamma_{1,\alpha}^{2}(\boldsymbol{x})\right)  \boldsymbol{x}%
\boldsymbol{x}^{T}dG(\boldsymbol{x})$ and $\boldsymbol{J}_{\alpha
}(\boldsymbol{\eta})={\int_{\mathcal{X}}}\gamma_{11,\alpha}(\boldsymbol{x}%
)\boldsymbol{x}\boldsymbol{x}^{T}dG(\boldsymbol{x})$.

\subsection{Robustness Properties: Influence Function}

\label{sec2.2}

Let us now study the robustness of the MDPDEs $\widehat{\boldsymbol{\eta}%
}_{\alpha}$ of $\boldsymbol{\eta}$ through the classical influence function of
Hampel et al.~(1986). Let us rewrite the MDPDE in terms of a statistical
functional $T_{\alpha}(H)$ at the true joint distribution $H(y,\boldsymbol{x}%
)$ of $(Y, \textbf{X})$ as the solution of (\ref{EQ:func_est}), whenever it exists.
Consider the contaminated distribution $H_{\epsilon}= (1-\epsilon)H +
\epsilon\wedge_{(y_{t},\boldsymbol{x}_{t})}$, where $\epsilon$ is the
contamination proportion and $\wedge_{(y_{t},\boldsymbol{x}_{t})}$ is the
degenerate distribution at the contamination point ${(y_{t},\boldsymbol{x}%
_{t})}$. Then, the influence function of $T_{\alpha}(H)$ is defined as
\begin{align}
\mathcal{IF}((y_{t},\boldsymbol{x}_{t}),T_{\alpha},H) = \left.  \frac{\partial
T_{\alpha}(H_{\epsilon})}{\partial\epsilon}\right\vert _{\epsilon=0} =
\lim\limits_{\epsilon\downarrow0} \frac{T_{\alpha}(H_{\epsilon}) - T_{\alpha
}(H)}{\epsilon},
\end{align}
which measures the bias in the estimator due to an infinitesimal
contamination in the data generating distribution.
Thus, a bounded influence function indicates local stability in the estimators
in terms of bounding the bias under contamination, which is referred to as (local) B-robustness.
Although there are several other important robustness measures as briefly pointed out later in Section 8,
throughout the present paper we will indicate such local B-robustness whenever we talk about robustness
of our MDPDE and the corresponding tests in terms of having a bounded influence function.

Note that, the MDPDE functional $T_{\alpha}(H)$ is clearly an M-estimator
functional and we can get its influence function directly from existing
M-estimator theory.
In particular, the influence function of the MDPDE functional $T_{\alpha}$ at
the model distribution $H_{0}(y,\boldsymbol{x})=f(y,\boldsymbol{x}%
^{T}\boldsymbol{\beta}, \phi)G(\boldsymbol{x})$ is given by
\begin{align}
\mathcal{IF}((y_{t},\boldsymbol{x}_{t}),T_{\alpha},H_{0})  &  = J_{\alpha
}(\boldsymbol{\eta})^{-1} \Psi_{\alpha}(y_{t},\boldsymbol{x}_{t}%
^{T}\boldsymbol{\beta},\phi)\nonumber\\
&  = J_{\alpha}(\boldsymbol{\eta})^{-1}\left(
\begin{array}
[c]{c}%
\left(  \gamma_{1,\alpha}(\boldsymbol{x}_{t})- K_{1}(y_{t},\boldsymbol{x}%
_{t}^{T}\boldsymbol{\beta},\phi)f^{\alpha}(y_{t},\boldsymbol{x}_{t}%
^{T}\boldsymbol{\beta},\phi)\right)  \boldsymbol{x}_{i}\\
\gamma_{2,\alpha}(\boldsymbol{x}_{t})- K_{2}(y_{t},\boldsymbol{x}_{t}%
^{T}\boldsymbol{\beta},\phi)f^{\alpha}(y_{t},\boldsymbol{x}_{t}^{T}%
\boldsymbol{\beta},\phi)
\end{array}
\right)
\end{align}
where $J_{\alpha}(\eta)$ is as defined in Section 2.1 and $(y_{t}%
,\boldsymbol{x}_{t})$ is the point of contamination.

Further, suppose $T_{\alpha}^{\beta}(H)$ and $T_{\alpha}^{\phi}(H)$ refer to
the MDPDE functionals corresponding to the parameters $\boldsymbol{\beta}$ and
$\phi$, respectively, so that $T_{\alpha}(H)= (T_{\alpha}^{\beta}(H)^{T},
T_{\alpha}^{\phi}(H))^{T}$. Note that the influence functions of the two
estimators $T_{\alpha}^{\beta}(H)$ and $T_{\alpha}^{\phi}(H)$ are not
independent in general linear models. However, whenever the matrix $J_{\alpha
}(\boldsymbol{\eta})$ is diagonal (as in the normal linear model) or $\phi$ is
known (as in the logistic and Poisson regression models), the influence function of
the MDPDE of $\boldsymbol{\beta}$ can be written simply as
\begin{align}
\mathcal{IF}((y_{t},\boldsymbol{x}_{t}),T_{\alpha}^{\beta},H_{0})  &  =
\left(  {\int_{\mathcal{X}}}\gamma_{11,\alpha}(\boldsymbol{x} )\boldsymbol{x}%
\boldsymbol{x}^{T}dG(\boldsymbol{x})\right)  ^{-1} \left(  \gamma_{1,\alpha
}(\boldsymbol{x}_{t})- K_{1}(y_{t},\boldsymbol{x}_{t}^{T}\boldsymbol{\beta
},\phi)f^{\alpha}(y_{t},\boldsymbol{x}_{t}^{T}\boldsymbol{\beta},\phi)\right)
\boldsymbol{x}_{i}.
\end{align}

From the above form it is easily observed that this  influence function  is
bounded in the contamination point $(y_{t},\boldsymbol{x}_{t})$ for all $\alpha>0$
and unbounded at $\alpha=0$ for most standard GLMs. For example, under the
 normal linear regression model, the influence function of the MDPDE of
$\beta$ depends on the contamination point through the quantity $(y_{t}%
-\boldsymbol{x}_{t}^{T}\boldsymbol{\beta})\boldsymbol{x}_{t}e^{-\frac
{\alpha(y_{t}-x_{t}^{T}\boldsymbol{\beta})^{2}}{2\phi^{2}}}$ and hence it is
bounded for all $\alpha>0$ implying the robustness of the corresponding MDPDEs. 
In this paper we will present the general theory of the random design model, 
and illustrate the methodology in detail for the Poisson regression problem.

\section{Wald-type Test Statistics for General Composite Hypothesis
\label{sec3}}

The asymptotic distribution of $\widehat{\boldsymbol{\eta}}_{\alpha}$, given
in Section \ref{sec2.1}, will be useful in order to define a family of
Wald-type test statistics for testing the null hypothesis
\begin{equation}
H_{0}:\boldsymbol{m}\left(  \boldsymbol{\eta}\right)  =\boldsymbol{0}\text{
versus }H_{1}:\boldsymbol{m}\left(  \boldsymbol{\eta}\right)  \neq
\boldsymbol{0}, \label{EQ:hyp}%
\end{equation}
with $\boldsymbol{m}:\mathbb{R}^{k+1}\rightarrow\mathbb{R}^{r},$ $r<k+1$. Thus the null hypothesis imposes $r$ restrictions on the parameter $\eta$.  We
shall assume that $\boldsymbol{M}\left(  \boldsymbol{\eta}\right)
=\frac{\partial}{\partial\boldsymbol{\eta}}\boldsymbol{m}^{T}\left(
\boldsymbol{\eta}\right)  $ is a continuous full (column) rank matrix with $k+1$ rows
and $r$ columns.

If $\phi$ is known or we are only interested in testing some hypothesis on
$\boldsymbol{\beta}$, say, $\boldsymbol{m}^{\ast}\left( \boldsymbol{\beta
}\right) =\boldsymbol{0}$, we shall consider $\boldsymbol{m}\left(
\boldsymbol{\eta}\right) =\boldsymbol{m}^{\ast}\left( \boldsymbol{\beta
}\right) $ and then $\boldsymbol{M}\left( \boldsymbol{\eta}\right)  = \left(
\frac{\partial}{\partial\boldsymbol{\beta}^{T}}\boldsymbol{m}^{\ast}\left(
\boldsymbol{\beta}\right) , \ \boldsymbol{0}\right) ^{T}$ if $\phi$ is
unknown, and $\boldsymbol{M}\left( \boldsymbol{\eta}\right)  = \frac{\partial
}{\partial\boldsymbol{\beta}^T}\boldsymbol{m}^{\ast}\left( \boldsymbol{\beta
}\right) $ if $\phi$ is known. The most commonly used hypothesis under this
set-up is the general linear hypothesis on $\boldsymbol{\beta}$ given by
$\boldsymbol{L}\boldsymbol{\beta}=\boldsymbol{l}_{0}$ for some $r\times k$
matrix $\boldsymbol{L}$ and $r$-vector $\boldsymbol{l}_{0}$. Here we 
have $\boldsymbol{m}\left( \boldsymbol{\eta}\right) =\boldsymbol{m}^{\ast
}\left( \boldsymbol{\beta}\right)  =\boldsymbol{L}\boldsymbol{\beta
}-\boldsymbol{l}_{0}$ and $\boldsymbol{M}\left( \boldsymbol{\eta}\right)  =
\left( \boldsymbol{L}, \ 0\right) ^{T}$ or $\boldsymbol{L}^{T}$ for $\phi$
unknown or known respectively. On the other hand, if we are interested in
testing $H_{0}:\phi=\phi_{0}$, we shall consider $\boldsymbol{m}\left(
\boldsymbol{\eta}\right)  =\phi-\phi_{0}.$ In this case $\boldsymbol{M}\left(
\boldsymbol{\eta}\right)  =\left(  \boldsymbol{0}_{1\times k}^{T},1\right)
^{T}$.

\begin{definition}
Let $\widehat{\boldsymbol{\eta}}_{\alpha}$ be the MDPDE for $\boldsymbol{\eta
}$. The family of Wald-type test statistics for testing the null hypothesis
given in (\ref{EQ:hyp}) is given by
\begin{equation}
W_{n}(\widehat{\boldsymbol{\eta}}_{\alpha})=n\boldsymbol{m}\left(
\widehat{\boldsymbol{\eta}}_{\alpha}\right)  ^{T}\left[  \boldsymbol{M}%
(\widehat{\boldsymbol{\eta}}_{\alpha})^{T}\boldsymbol{\Sigma}_{\alpha
}(\widehat{\boldsymbol{\eta}}_{\alpha})\boldsymbol{M}%
(\widehat{\boldsymbol{\eta}}_{\alpha})\right]  ^{-1}\boldsymbol{m}\left(
\widehat{\boldsymbol{\eta}}_{\alpha}\right)  . \label{EQ:TS}%
\end{equation}

\end{definition}

\begin{theorem}
The asymptotic distribution of the Wald-type test statistic, $W_{n}%
(\widehat{\boldsymbol{\eta}}_{\alpha})$, defined in (\ref{EQ:TS}), under the
null hypothesis given in (\ref{EQ:hyp}), is a chi-square distribution with $r$
degrees of freedom. \label{THM:asymp_null}
\end{theorem}

\begin{proof}
We know that $\sqrt{n}(\widehat{\boldsymbol{\eta}}_{\alpha}-\boldsymbol{\eta
}_{0})\underset{n\rightarrow\infty}{\overset{\mathcal{L}}{\longrightarrow}%
}\mathcal{N}\left(  \boldsymbol{0}\text{,}\boldsymbol{\Sigma}_{\alpha}\left(
\boldsymbol{\eta}_{0}\right)  \right)  $ and $\boldsymbol{m}\left(
\widehat{\boldsymbol{\eta}}_{\alpha}\right)  =\boldsymbol{M}(\boldsymbol{\eta
}_{0})^{T}(\widehat{\boldsymbol{\eta}}_{\alpha}-\boldsymbol{\eta}_{0}%
)+o_{p}(n^{-1/2})$ because $\boldsymbol{m}\left(  \boldsymbol{\eta}%
_{0}\right)  =\boldsymbol{0}$. Therefore
\[
\sqrt{n} \boldsymbol{m}\left(  \widehat{\boldsymbol{\eta}}_{\alpha}\right)
\underset{n\rightarrow\infty}{\overset{\mathcal{L}}{\longrightarrow}%
}\mathcal{N}\left(  \boldsymbol{0},\boldsymbol{M}(\boldsymbol{\eta}_{0}%
)^{T}\boldsymbol{\boldsymbol{\Sigma}_{\alpha}}(\boldsymbol{\boldsymbol{\eta
}_{0}})\boldsymbol{M}(\boldsymbol{\eta}_{0})\right)  .
\]
Then the asymptotic distribution of $W_{n}(\widehat{\boldsymbol{\eta}}%
_{\alpha})$ is a chi-square distribution with $r$ degrees of freedom.
\end{proof}

Based on the previous theorem the null hypothesis given in (\ref{EQ:hyp}) will
be rejected at $\alpha_0$ if we have 
\begin{equation}
W_{n}(\widehat{\boldsymbol{\eta}}_{\alpha})>\chi_{r,\alpha_0}^{2}. \label{3.111}%
\end{equation}
Now we consider $\boldsymbol{\eta}^{\ast}$ $\in\Theta$ such that
$\boldsymbol{m}\left(  \boldsymbol{\eta}^{\ast}\right)  \neq\boldsymbol{0}$,
i.e., $\boldsymbol{\eta}^{\ast}$ does not belong to the null hypothesis. We
denote%
\[
q_{\boldsymbol{\eta}_{1}}(\boldsymbol{\eta}_{2})=m\left(  \boldsymbol{\eta
}_{1}\right)  ^{T}\left(  \boldsymbol{\boldsymbol{M}^{T}}\left(
\boldsymbol{\eta}_{2}\right)  \boldsymbol{\boldsymbol{\Sigma}_{\alpha}%
}(\boldsymbol{\boldsymbol{\eta}_{2}})\boldsymbol{M}\left(  \boldsymbol{\eta
}_{2}\right)  \right)  ^{-1}m\left(  \boldsymbol{\eta}_{1}\right)
\]
and, in the following, we provide  an approximation to the power function for the Wald-type tests
given in (\ref{3.111}).

\begin{theorem}
Let $\boldsymbol{\eta}^{\ast}\in\Theta$ be the true value of the parameter
such that $\boldsymbol{m}\left(  \boldsymbol{\eta}^{\ast}\right)
\neq\boldsymbol{0}$ and $\widehat{\boldsymbol{\eta}}_{\alpha}%
\underset{n\rightarrow\infty}{\overset{P}{\longrightarrow}}\boldsymbol{\eta
}^{\ast}$. The power function of the tests given in (\ref{3.111}), in
$\boldsymbol{\eta}^{\ast}$, is given by
\begin{equation}
\pi\left(  \boldsymbol{\eta}^{\ast}\right)  =1-\Phi_{n}\left(  \frac{1}%
{\sigma\left(  \boldsymbol{\eta}^{\ast}\right)  }\left(  \frac{\chi_{r,\alpha_0
}^{2}}{\sqrt{n}}-\sqrt{n}q_{\boldsymbol{\eta}^{\ast}}(\boldsymbol{\eta}^{\ast
})\right)  \right)  \label{3.112}%
\end{equation}
where $\Phi_{n}\left(  x\right)  $ almost surely converges to the standard normal
distribution $\Phi\left(  x\right)  $ and $\sigma\left(  \boldsymbol{\eta
}^{\ast}\right)  $ is given by
\[
\sigma^{2}\left(  \boldsymbol{\eta}^{\ast}\right)  =\left.  \frac{\partial
q_{\boldsymbol{\eta}}(\boldsymbol{\eta})}{\partial\boldsymbol{\eta}^{T}%
}\right\vert _{\boldsymbol{\eta}=\boldsymbol{\eta}^{\ast}}\boldsymbol{\Sigma
}_{\alpha}\left(  \boldsymbol{\eta}_{0}\right)  \left.  \frac{\partial
q_{\boldsymbol{\eta}}(\boldsymbol{\eta})}{\partial\boldsymbol{\eta}%
}\right\vert _{\boldsymbol{\eta}=\boldsymbol{\eta}^{\ast}}.
\]
\label{THM:power_approx}
\end{theorem}

\begin{proof}
We have
\begin{align*}
\pi\left(  \boldsymbol{\eta}^{\ast}\right)   &  =\Pr\left(  W_{n}%
(\widehat{\boldsymbol{\eta}}_{\alpha})>\chi_{r,\alpha_0}^{2}\right)  =\Pr\left(
n\left(  q_{\widehat{\boldsymbol{\eta}}_{\alpha}}(\widehat{\boldsymbol{\eta}%
}_{\alpha})-q_{\boldsymbol{\eta}^{\ast}}(\boldsymbol{\eta}^{\ast})\right)
>\chi_{r,\alpha_0}^{2}-nq_{\boldsymbol{\eta}^{\ast}}(\boldsymbol{\eta}^{\ast
})\right) \\
&  =\Pr\left(  \sqrt{n}\left(  q_{\widehat{\boldsymbol{\eta}}_{\alpha}%
}(\widehat{\boldsymbol{\eta}}_{\alpha})-q_{\boldsymbol{\eta}^{\ast}%
}(\boldsymbol{\eta}^{\ast})\right)  >\frac{\chi_{r,\alpha_0}^{2}}{\sqrt{n}%
}-\sqrt{n}q_{\boldsymbol{\eta}^{\ast}}(\boldsymbol{\eta}^{\ast})\right)  .
\end{align*}
Now we are going to get the asymptotic distribution of the random variable
$\sqrt{n}\left(  q_{\widehat{\boldsymbol{\eta}}_{\alpha}}%
(\widehat{\boldsymbol{\eta}}_{\alpha})-q_{\boldsymbol{\eta}^{\ast}%
}(\boldsymbol{\eta}^{\ast})\right)  $. Since $\widehat{\boldsymbol{\eta}%
}_{\alpha}\underset{n\rightarrow\infty}{\overset{P}{\longrightarrow}%
}\boldsymbol{\eta}^{\ast}$, it is clear that $q_{\widehat{\boldsymbol{\eta}%
}_{\alpha}}(\widehat{\boldsymbol{\eta}}_{\alpha})$ and
$q_{\widehat{\boldsymbol{\eta}}_{\alpha}}(\boldsymbol{\eta}^{\ast})$ have the
same asymptotic distribution. The first order Taylor expansion of
$q_{{\boldsymbol{\eta}}}(\boldsymbol{\eta}^{\ast})$ around $\boldsymbol{\eta}=\boldsymbol{\eta}^{\ast}$,
evaluated at $\boldsymbol{\eta}=\widehat{\boldsymbol{\eta}}_{\alpha}$,  
gives%
\[
q_{\widehat{\boldsymbol{\eta}}_{\alpha}}(\boldsymbol{\eta}^{\ast
})-q_{\boldsymbol{\eta}^{\ast}}(\boldsymbol{\eta}^{\ast})=\left.
\frac{\partial q_{\boldsymbol{\eta}}(\boldsymbol{\eta}^{\ast})}{\partial
\boldsymbol{\eta}^{T}}\right\vert _{\boldsymbol{\eta}=\boldsymbol{\eta}^{\ast
}}(\widehat{\boldsymbol{\eta}}_{\alpha}-\boldsymbol{\eta}^{\ast})+o_{p}\left(
\left\Vert \widehat{\boldsymbol{\eta}}_{\alpha}-\boldsymbol{\eta}^{\ast
}\right\Vert \right)  .
\]
Therefore, it holds%
\[
\sqrt{n}\left(  q_{\widehat{\boldsymbol{\eta}}_{\alpha}}%
(\widehat{\boldsymbol{\eta}}_{\alpha})-q_{\boldsymbol{\eta}^{\ast}%
}(\boldsymbol{\eta}^{\ast})\right)  \underset{n\rightarrow\infty
}{\overset{\mathcal{L}}{\longrightarrow}}\mathcal{N}\left(  0,\sigma
^{2}\left(  \boldsymbol{\eta}^{\ast}\right)  \right)  ,
\]
and the result follows.
\end{proof}

\begin{remark}
Based on the previous theorem, we can obtain the sample size $n$ necessary to get
a specific power $\pi\left(  \boldsymbol{\eta}^{\ast}\right)  =\pi_{0}$. From
(\ref{3.112}), we must solve the equation
\[
1-\pi_{0}=\Phi\left(  \frac{1}{\sigma\left(  \boldsymbol{\eta}^{\ast}\right)
}\left(  \frac{\chi_{r,\alpha_0}^{2}}{\sqrt{n}}-\sqrt{n}q_{\boldsymbol{\eta
}^{\ast}}(\boldsymbol{\eta}^{\ast})\right)  \right)
\]
and we get that $n=\left[  n^{\ast}\right]  +1$ with
\[
n^{\ast}=\frac{A+B+\sqrt{A(A+2B)}}{2q_{\boldsymbol{\eta}^{\ast}}%
^{2}(\boldsymbol{\eta}^{\ast})},
\]
where
\[
A=\sigma^{2}\left(  \boldsymbol{\eta}^{\ast}\right)  \left(  \Phi^{-1}\left(
1-\pi_{0}\right)  \right)  ^{2}\text{ and }B=2\chi_{r,\alpha_0}^{2}%
q_{\boldsymbol{\eta}^{\ast}}(\boldsymbol{\eta}^{\ast}).
\]

\end{remark}

\begin{corollary}
Under the assumptions of Theorem \ref{THM:power_approx}, we have $\pi\left(
\boldsymbol{\eta}^{\ast}\right)  \rightarrow1$ as $n\rightarrow\infty.$ Thus,
our proposed Wald-type tests are consistent at any fixed alternative.
\label{CORR:consistency}
\end{corollary}

We may also find an approximation of the power of the Wald-type tests given in (\ref{EQ:TS}) at an alternative close to the null
hypothesis. Let $\boldsymbol{\eta}_{n}\in\Theta-\Theta_{0}$ be a given
alternative and let $\boldsymbol{\eta}_{0}$ be the element in boundary of $\Theta_{0}$
closest to $\boldsymbol{\eta}_{n}$ in the Euclidean distance sense. One
possibility to introduce contiguous alternative hypotheses in this set up  is to consider a
fixed vector $\boldsymbol{d}$ and to permit $\boldsymbol{\eta}_{n}$
to move towards $\boldsymbol{\eta}_{0}$ with increasing $n$ as 
\begin{equation}
H_{1,n}:\boldsymbol{\eta}_{n}=\boldsymbol{\eta}_{0}+n^{-1/2}\boldsymbol{d}.
\label{a.15}%
\end{equation}
A second approach could be to relax the condition $\boldsymbol{m}\left(
\boldsymbol{\eta}\right)  =\boldsymbol{0}$ defining $\Theta_{0}.$ Let
$\boldsymbol{d}^{\ast}\in\mathbb{R}^{r}$ and consider the  sequence
$\left\{  \boldsymbol{\eta}_{n}\right\}  $ of parameters moving towards
$\boldsymbol{\eta}_{0}$ according to
\begin{equation}
H_{1,n}^{\ast}:\boldsymbol{m}\left(  \boldsymbol{\eta}_{n}\right)
=n^{-1/2}\boldsymbol{d}^{\ast}. \label{a.16}%
\end{equation}
Note that a Taylor series expansion of $\boldsymbol{m}\left(  \boldsymbol{\eta
}_{n}\right)  $ around $\boldsymbol{\eta}_{0}$ yields
\begin{equation}
\boldsymbol{m}\left(  \boldsymbol{\eta}_{n}\right)  =\boldsymbol{m}%
(\boldsymbol{\eta}_{0})+\boldsymbol{M}^{T}(\boldsymbol{\eta}_{0})\left(
\boldsymbol{\eta}_{n}-\boldsymbol{\eta}_{0}\right)  +o\left(  \left\Vert
\boldsymbol{\eta}_{n}-\boldsymbol{\eta}_{0}\right\Vert \right)  . \label{a.17}%
\end{equation}
By substituting $\boldsymbol{\eta}_{n}=\boldsymbol{\eta}_{0}+n^{-1/2}%
\boldsymbol{d}$ in (\ref{a.17}) and taking into account that\textbf{\ }%
$\boldsymbol{m}(\boldsymbol{\eta}_{0})=\boldsymbol{0}$, we get
\[
\boldsymbol{m}\left(  \boldsymbol{\eta}_{n}\right)  =n^{-1/2}\boldsymbol{M}%
^{T}(\boldsymbol{\eta}_{0})\boldsymbol{d}+o\left(  \left\Vert \boldsymbol{\eta
}_{n}-\boldsymbol{\eta}_{0}\right\Vert \right)  ,
\]
so that the equivalence of the two approaches  in the limit is obtained for $\boldsymbol{d}^{\ast
}\boldsymbol{=M}^{T}(\boldsymbol{\eta}_{0})\boldsymbol{d}$.

In the following we shall denote by $\chi_{l}^{2}(s)$ the non-central
chi-square random variable with $l$ degrees of freedom and non-centrality
parameter $s$.

\begin{theorem}
\label{THM:Asump_contg} We have the following results under both versions of
the contiguous alternative hypothesis:

\begin{enumerate}
\item[i)] $W_{n}(\widehat{\boldsymbol{\eta}}_{\alpha})\underset{n\rightarrow
\infty}{\overset{\mathcal{L}}{\longrightarrow}}\chi_{r}^{2}\left( a\right)  $
under $H_{1,n}$ given in (\ref{a.15}).

\item[ii)] $W_{n}(\widehat{\boldsymbol{\eta}}_{\alpha})\underset{n\rightarrow
\infty}{\overset{\mathcal{L}}{\longrightarrow}}\chi_{r}^{2}\left( b\right)  $
under $H_{1,n}^{\ast}$ given in (\ref{a.16}),
\end{enumerate}
where $a = \boldsymbol{d}^{T}\boldsymbol{M}(\boldsymbol{\eta}_{0})\left(
\boldsymbol{M}^{T}(\boldsymbol{\eta}_{0})\boldsymbol{\Sigma}_{\alpha
}(\boldsymbol{\eta}_{0})\boldsymbol{M}(\boldsymbol{\eta}_{0})\right)
^{-1}\boldsymbol{M}^{T}(\boldsymbol{\eta}_{0})\boldsymbol{d}$ and $b =
\boldsymbol{\boldsymbol{d}}^{\ast T}\left(  \boldsymbol{M}^{T}%
(\boldsymbol{\eta}_{0})\boldsymbol{\Sigma}_{\alpha}(\boldsymbol{\eta}%
_{0})\boldsymbol{M}(\boldsymbol{\eta}_{0})\right)  ^{-1}%
\boldsymbol{\boldsymbol{d}}^{\ast}$.

\begin{proof}
A Taylor series expansion of $\boldsymbol{m}(\widehat{\boldsymbol{\eta}%
}_{\beta})$ around $\boldsymbol{\eta}_{n}$ yields%
\[
\boldsymbol{m}(\widehat{\boldsymbol{\eta}}_{\beta})=\boldsymbol{m}\left(
\boldsymbol{\eta}_{n}\right)  +\boldsymbol{M}^{T}\left(  \boldsymbol{\eta}%
_{n}\right)  (\widehat{\boldsymbol{\eta}}_{\beta}-\boldsymbol{\eta}%
_{n})+o\left(  \left\Vert \widehat{\boldsymbol{\eta}}_{\beta}-\boldsymbol{\eta
}_{n}\right\Vert \right)  .
\]
From (\ref{a.17}), we have
\[
\boldsymbol{m}(\widehat{\boldsymbol{\eta}}_{\beta})=\boldsymbol{M}%
^{T}(\boldsymbol{\eta}_{0})n^{-1/2}\boldsymbol{d}+\boldsymbol{M}^{T}\left(
\boldsymbol{\eta}_{n}\right)  (\widehat{\boldsymbol{\eta}}_{\beta
}-\boldsymbol{\eta}_{n})+o\left(  \left\Vert \widehat{\boldsymbol{\eta}%
}_{\beta}-\boldsymbol{\eta}_{n}\right\Vert \right)  +o\left(  \left\Vert
\boldsymbol{\eta}_{n}-\boldsymbol{\eta}_{0}\right\Vert \right)  .
\]
As $\sqrt{n}\left(  o\left(  \left\Vert \widehat{\boldsymbol{\eta}}_{\beta
}-\boldsymbol{\eta}_{n}\right\Vert \right)  +o\left(  \left\Vert
\boldsymbol{\eta}_{n}-\boldsymbol{\eta}_{0}\right\Vert \right)  \right)
=o_{p}\left(  1\right)  $ and%
\[
\sqrt{n}(\widehat{\boldsymbol{\eta}}_{\beta}-\boldsymbol{\eta}_{n}%
)\underset{n\rightarrow\infty}{\overset{\mathcal{L}}{\longrightarrow}%
}\mathcal{N}(\boldsymbol{0},\boldsymbol{\Sigma}_{\alpha}(\boldsymbol{\eta}%
_{0})).
\]
we have
\[
\sqrt{n}\boldsymbol{m}(\widehat{\boldsymbol{\eta}}_{\beta}%
)\underset{n\rightarrow\infty}{\overset{\mathcal{L}}{\longrightarrow}%
}\mathcal{N}(\boldsymbol{M}^{T}(\boldsymbol{\eta}_{0})\boldsymbol{d}%
,\boldsymbol{M}^{T}(\boldsymbol{\eta}_{0})\boldsymbol{\Sigma}_{\alpha
}(\boldsymbol{\eta}_{0})\boldsymbol{M}(\boldsymbol{\eta}_{0})).
\]
We can observe by the relationship $\boldsymbol{d}^{\ast}\boldsymbol{=M}%
^{T}(\boldsymbol{\eta}_{0})\boldsymbol{d}$, if $\boldsymbol{m}\left(
\boldsymbol{\eta}_{n}\right)  =n^{-1/2}\boldsymbol{d}^{\ast}$ that
\[
\sqrt{n}\boldsymbol{m}(\widehat{\boldsymbol{\eta}}_{\beta}%
)\underset{n\rightarrow\infty}{\overset{\mathcal{L}}{\longrightarrow}%
}\mathcal{N}(\boldsymbol{d}^{\ast},\boldsymbol{M}^{T}(\boldsymbol{\eta}%
_{0})\boldsymbol{\Sigma}_{\alpha}(\boldsymbol{\eta}_{0})\boldsymbol{M}%
(\boldsymbol{\eta}_{0})).
\]
We apply the following result from Anderson (2003) concerning quadratic forms. \textquotedblleft If
$\boldsymbol{Z\sim}\mathcal{N}\left(  \boldsymbol{\mu},\boldsymbol{\Sigma
}\right)  $, $\boldsymbol{\Sigma}$ is a symmetric projection of rank $k$ and
$\boldsymbol{\Sigma\mu=\mu}$, then $\boldsymbol{Z}^{T}\boldsymbol{Z}$ is a
chi-square distribution with $k$ degrees of freedom and noncentrality
parameter $\boldsymbol{\mu}^{T}\boldsymbol{\mu}$\textquotedblright.  In our
case, the quadratic form is
\[
W_{n}=\boldsymbol{Z}^{T}\boldsymbol{Z}%
\]
with
\[
\boldsymbol{Z}=\sqrt{n}\left(  \boldsymbol{M}^{T}(\boldsymbol{\eta}%
_{0})\boldsymbol{\Sigma}_{\alpha}(\boldsymbol{\eta}_{0})\boldsymbol{M}%
(\boldsymbol{\eta}_{0})\right)  ^{-1/2}\boldsymbol{m}%
(\widehat{\boldsymbol{\eta}}_{\beta})
\]
and
\[
\boldsymbol{Z}\underset{n\rightarrow\infty}{\overset{\mathcal{L}%
}{\longrightarrow}}\mathcal{N}\left(  \left(  \boldsymbol{M}^{T}%
(\boldsymbol{\eta}_{0})\boldsymbol{\Sigma}_{\alpha}(\boldsymbol{\eta}%
_{0})\boldsymbol{M}(\boldsymbol{\eta}_{0})\right)  ^{-1/2}\boldsymbol{M}%
^{T}(\boldsymbol{\eta}_{0})\boldsymbol{d},\boldsymbol{I}\right)  ,
\]
where $\boldsymbol{I}$ is the $r\times r$ identity matrix. Hence, the
application of the result is immediate and the noncentrality parameter is
\[
\boldsymbol{d}^{T}\boldsymbol{M}(\boldsymbol{\eta}_{0})\left(  \boldsymbol{M}%
^{T}(\boldsymbol{\eta}_{0})\boldsymbol{\Sigma}_{\alpha}(\boldsymbol{\eta}%
_{0})\boldsymbol{M}(\boldsymbol{\eta}_{0})\right)  ^{-1}\boldsymbol{M}%
^{T}(\boldsymbol{\eta}_{0})\boldsymbol{d}=\boldsymbol{d}^{\ast T}\left(
\boldsymbol{M}^{T}(\boldsymbol{\eta}_{0})\boldsymbol{\Sigma}_{\alpha
}(\boldsymbol{\eta}_{0})\boldsymbol{M}(\boldsymbol{\eta}_{0})\right)
^{-1}\boldsymbol{d}^{\ast}.
\]

\end{proof}
\end{theorem}

\section{Robustness of the Proposed Wald-type Test Statistics \label{sec4}}

\subsection{Influence Function of the Wald-type Test Statistics}

In order to study the robustness of the proposed Wald-type tests of Section
\ref{sec3}, we will start with the influence function of the Wald-type test
statistics $W_{n}(\widehat{\boldsymbol{\eta}}_{\alpha})$ in (\ref{EQ:TS}) for
testing the general composite hypothesis (\ref{EQ:hyp}). Consider the MDPDE
functional $T_{\alpha}(H)$ at the true joint distribution $H$ of $(Y,
\boldsymbol{X})$ as defined in Section \ref{sec2.2} and define the statistical
functional corresponding to the Wald-type test statistics $W_{n}(\widehat{\boldsymbol{\eta}}_{\alpha})$ at $H$ as
(ignoring the multiplier $n$)
\begin{equation}
W_{\alpha}(H)=\boldsymbol{m}\left(  T_{\alpha}(H)\right)  ^{T} \left[
\boldsymbol{M}\left(  \boldsymbol{\eta}\right)  ^{T}\boldsymbol{\Sigma
}_{\alpha}\left(  \boldsymbol{\eta}\right)  \boldsymbol{M}\left(
\boldsymbol{\eta}\right)  \right]  ^{-1} \boldsymbol{m}\left(  T_{\alpha
}(H)\right)  . \label{WoF}%
\end{equation}

Again, considering the contaminated distribution, $H_{\epsilon}$, the
influence function of the Wald-type test functional $W_{\alpha}(\cdot)$ is
given by
\begin{align}
\mathcal{IF}((y_{t},\boldsymbol{x}_{t}),W_{\alpha},H)  &  = \left.
\frac{\partial W_{\alpha}(H_{\epsilon})}{\partial\epsilon}\right\vert
_{\epsilon=0}\nonumber\\
&  = \boldsymbol{m}\left(  T_{\alpha}(H)\right)  ^{T} \left[  \boldsymbol{M}%
\left(  \boldsymbol{\eta}\right)  ^{T}\boldsymbol{\Sigma}_{\alpha}\left(
\boldsymbol{\eta}\right)  \boldsymbol{M}\left(  \boldsymbol{\eta}\right)
\right]  ^{-1} \boldsymbol{M}\left(  \boldsymbol{\eta}\right)  ^{T}%
\mathcal{IF}((y_{t},\boldsymbol{x}_{t}),T_{\alpha},H).\nonumber
\end{align}
Suppose $\boldsymbol{\eta}_{0} = (\boldsymbol{\beta}_{0}, \phi_{0})$ be the
true parameter value under null hypothesis given in (\ref{EQ:hyp}) that
satisfies $\boldsymbol{m}(\boldsymbol{\eta}_{0}) = 0$ and the corresponding
null joint distribution be $H_{0}(y,\boldsymbol{x})=f(y,\boldsymbol{x}%
^{T}\boldsymbol{\beta}_{0}, \phi_{0})G(\boldsymbol{x})$. Note that, under
$H_{0}$, $T_{\alpha}(H_{0}) = \boldsymbol{\eta}_{0}$ by Fisher consistency of
the MDPDE and hence $\mathcal{IF}((y_{t},\boldsymbol{x}_{t}),W_{\alpha}%
,H_{0})=0$. Hence, the first order influence function cannot portray the
robustness of the proposed Wald-type tests (like other Wald-type tests in
Rousseeuw and Ronchetti, 1979; Toma and Broniatowski, 2011; Ghosh et al., 2016,
etc.) and we need to derive its second order influence function.

By another differentiation, we get the second order influence function of
$W_{\alpha}(\cdot)$ at $H$ as given by
\begin{align}
&  \mathcal{IF}_{2}((y_{t},\boldsymbol{x}_{t}),W_{\alpha},H) =\left.
\frac{\partial^{2} W_{\alpha}(H_{\epsilon})}{\partial\epsilon^{2}}\right\vert
_{\epsilon=0}\nonumber\\
&  = \boldsymbol{m}\left(  T_{\alpha}(H)\right)  ^{T} \left[  \boldsymbol{M}%
^{T}\left(  \boldsymbol{\eta}\right)  \boldsymbol{\Sigma}_{\alpha}\left(
\boldsymbol{\eta}\right)  \boldsymbol{M}\left(  \boldsymbol{\eta}\right)
\right]  ^{-1} \boldsymbol{M}^{T}\left(  \boldsymbol{\eta}\right)
\mathcal{IF}_{2}((y_{t},\boldsymbol{x}_{t}),T_{\alpha},H)\nonumber\\
&  + \mathcal{IF}((y_{t},\boldsymbol{x}_{t}),T_{\alpha},H)^{T}\boldsymbol{M}%
\left(  \boldsymbol{\eta}\right)  \left[  \boldsymbol{M}^{T}\left(
\boldsymbol{\eta}\right) \boldsymbol{\Sigma}_{\alpha}\left(  \boldsymbol{\eta
}\right)  \boldsymbol{M}\left(  \boldsymbol{\eta}\right)  \right]
^{-1}\boldsymbol{M}^{T}\left(  \boldsymbol{\eta}\right)  \mathcal{IF}%
((y_{t},\boldsymbol{x}_{t}),T_{\alpha},H).\nonumber
\end{align}
Note that the influence function of the test statistic is directly related to the influence function of the corresponding estimator. 
In particular, at the null distribution $H_{0}(y,\boldsymbol{x})$, we get the
nonzero second order influence function indicating the robustness properties
of the proposed Wald-type test statistics. These are summarized in the
following theorem.

\begin{theorem}
\label{THM:second_IT_test} The influence functions of the proposed Wald-type
test statistics $W_{n}$ at the null distribution $H_{0}(y,\boldsymbol{x}%
)=f(y,\boldsymbol{x}^{T}\boldsymbol{\beta}_{0}, \phi_{0})G(\boldsymbol{x})$ is
given by
\begin{align}
&  \mathcal{IF}((y_{t},\boldsymbol{x}_{t}),W_{\alpha},H_{0}) = 0\nonumber\\
&  \mathcal{IF}_{2}((y_{t},\boldsymbol{x}_{t}),W_{\alpha},H_{0})\nonumber\\
&  = \mathcal{IF}((y_{t},\boldsymbol{x}_{t}),T_{\alpha},H_{0})^{T}%
\boldsymbol{M}\left(  \boldsymbol{\eta}_{0}\right)  \left[  \boldsymbol{M}%
^{T}\left(  \boldsymbol{\eta}_{0}\right)  \boldsymbol{\Sigma}_{\alpha}\left(
\boldsymbol{\eta}_{0}\right)  \boldsymbol{M}\left(  \boldsymbol{\eta}%
_{0}\right)  \right]  ^{-1} \boldsymbol{M}^{T}\left(  \boldsymbol{\eta}%
_{0}\right)  \mathcal{IF}((y_{t},\boldsymbol{x}_{t}),T_{\alpha},H_{0}%
)\nonumber\\
&  = \Psi_{\alpha}(y_{t},\boldsymbol{x}_{t}^{T}\boldsymbol{\beta},\phi)^{T}
J_{\alpha}(\boldsymbol{\eta})^{-1}\boldsymbol{M}\left(  \boldsymbol{\eta}%
_{0}\right)  \left[  \boldsymbol{M}^{T}\left(  \boldsymbol{\eta}_{0}\right)
\boldsymbol{\Sigma}_{\alpha}\left(  \boldsymbol{\eta}_{0}\right)
\boldsymbol{M}\left(  \boldsymbol{\eta}_{0}\right)  \right]  ^{-1}
\boldsymbol{M}^{T}\left(  \boldsymbol{\eta}_{0}\right)  J_{\alpha
}(\boldsymbol{\eta})^{-1} \Psi_{\alpha}(y_{t},\boldsymbol{x}_{t}%
^{T}\boldsymbol{\beta},\phi).\nonumber
\end{align}

\end{theorem}

Clearly, the second order influence function $\mathcal{IF}_2$ is bounded whenever the
function $\Psi_{\alpha}(y_{t},\boldsymbol{x}_{t}^{T}\boldsymbol{\beta},\phi)$
is bounded, i.e., for all $\alpha>0$, implying the robustness of the proposed
Wald-type tests with $\alpha>0$. However, at $\alpha=0$, $\Psi_{0}%
(y_{t},\boldsymbol{x}_{t}^{T}\boldsymbol{\beta},\phi)$ and hence the second
order influence function is unbounded implying the non-robust nature of the
classical MLE based Wald-test.

\subsection{Level and Power Robustness}

Let us now study the stability of the level and the power of the proposed
Wald-type test statistics under data contamination. For this, we will derive
the level and power influence functions respectively under the null hypothesis
and the contiguous alternative hypotheses $\boldsymbol{\eta}_{n}=\boldsymbol{\eta}%
_{0}+n^{-1/2}\boldsymbol{d}$ in (\ref{a.15}). Considering contamination over
these hypothesis as in Hampel et al.~(1986) and Ghosh et al.~(2016), we define
the LIF and PIF respectively through the asymptotic distribution under
\[
H_{n,\epsilon,(y_{t},\boldsymbol{x}_{t})}^{P}=(1-\frac{\epsilon}{\sqrt{n}%
})H_{\boldsymbol{\eta}_{n}}+\frac{\epsilon}{\sqrt{n}}\wedge_{(y_{t}%
,\boldsymbol{x}_{t})},~~~\mbox{ and } ~~H_{n,\epsilon,(y_{t},\boldsymbol{x}%
_{t})}^{L}=(1-\frac{\epsilon}{\sqrt{n}})H_{\boldsymbol{\eta}_{0}}%
+\frac{\epsilon}{\sqrt{n}}\wedge_{(y_{t},\boldsymbol{x}_{t})},
\]
where $H_{\boldsymbol{\eta}}$ denote the joint model distribution of
$(Y,\boldsymbol{X})$ with parameter $\boldsymbol{\eta}=(\boldsymbol{\beta}%
^{T}, \phi)^{T}$, given by $H_{\boldsymbol{\eta}}(y,\boldsymbol{x}%
)=f(y,\boldsymbol{x}^{T}\boldsymbol{\beta}, \phi)G(\boldsymbol{x})$. For the
proposed Wald-type test statistics $W_{n}$, its LIF and PIF are defined by
\[
\mathcal{LIF}((y_{t},\boldsymbol{x}_{t}); W_{n},H_{\boldsymbol{\eta}_{0}})
=\left.  \dfrac{\partial}{\partial\epsilon}\alpha(\epsilon,(y_{t}%
,\boldsymbol{x}_{t}))\right\vert _{\epsilon=0} = \left.  \dfrac{\partial
}{\partial\epsilon} \lim\limits_{n\rightarrow\infty}P_{H_{n,\epsilon
,(y_{t},\boldsymbol{x}_{t})}^{L}}(W_{n}>\chi_{r,\alpha_0}^{2}) \right\vert
_{\epsilon=0},
\]
and
\[
\mathcal{PIF}((y_{t},\boldsymbol{x}_{t}); W_{n},H_{\boldsymbol{\eta}_{0}})
=\left.  \dfrac{\partial}{\partial\epsilon}\pi(\boldsymbol{\eta}_{n}%
,\epsilon,(y_{t},\boldsymbol{x}_{t}))\right\vert _{\epsilon=0} =\left.
\dfrac{\partial}{\partial\epsilon} \lim\limits_{n\rightarrow\infty
}P_{H_{n,\epsilon,(y_{t},\boldsymbol{x}_{t})}^{P}}(W_{n}>\chi_{r,\alpha_0}^{2})
\right\vert _{\epsilon=0}.
\]
%

\begin{theorem}
\label{THM:7asymp_power_one} Under the assumptions of Theorem 5, we have the following:

\begin{enumerate}
\item Under $H_{n,\epsilon,(y_{t},\boldsymbol{x}_{t})}^{P}$, the proposed
Wald-type test statistics $W_{n}$ asymptotically follows a non-central
chi-square distribution with $r$ degrees of freedom and  non-centrality
parameter
\begin{align}
\delta= \widetilde{\boldsymbol{d}}_{\epsilon,(y_{t},\boldsymbol{x}_{t}%
),\alpha}(\boldsymbol{\eta}_{0})^{T} \boldsymbol{M}(\boldsymbol{\eta}_{0})
\left[  \boldsymbol{M}^{T}(\boldsymbol{\eta}_{0})\boldsymbol{\Sigma}_{\alpha
}\left(  \boldsymbol{\eta}_{0}\right)  \boldsymbol{M}(\boldsymbol{\eta}%
_{0})\right]  ^{-1}\boldsymbol{M}^{T}(\boldsymbol{\eta}_{0})
\widetilde{\boldsymbol{d}}_{\epsilon,(y_{t},\boldsymbol{x}_{t}),\alpha
}(\boldsymbol{\eta}_{0}),\label{EQ:delta}%
\end{align}
where $\widetilde{\boldsymbol{d}}_{\epsilon,(y_{t},\boldsymbol{x}_{t}),\alpha
}(\boldsymbol{\eta}_{0}) =\boldsymbol{d}+\epsilon\mathcal{IF}((y_{t}%
,\boldsymbol{x}_{t}),\boldsymbol{T}_{\alpha},H_{\boldsymbol{\eta}_{0}}).$

\item The asymptotic power function under $G_{n,\epsilon,(y_{t},\boldsymbol{x}%
_{t})}^{P}$ can be approximated as
\begin{align}
&  \pi(\boldsymbol{\eta}_{n},\epsilon,(y_{t},\boldsymbol{x}_{t})) =
\lim\limits_{n\rightarrow\infty}P_{H_{n,\epsilon,(y_{t},\boldsymbol{x}_{t}%
)}^{P}}(W_{n}>\chi_{r,\alpha_0}^{2})\nonumber\\
&  \cong\sum\limits_{v=0}^{\infty} C_{v}\left(  \boldsymbol{M}^{T}%
(\boldsymbol{\eta}_{0})\widetilde{\boldsymbol{d}}_{\epsilon,(y_{t}%
,\boldsymbol{x}_{t}),\alpha}(\boldsymbol{\beta}_{0}), \left[  \boldsymbol{M}%
^{T}(\boldsymbol{\eta}_{0})\boldsymbol{\Sigma}_{\alpha}\left(
\boldsymbol{\eta}_{0}\right)  \boldsymbol{M}(\boldsymbol{\eta}_{0})\right]
^{-1}\right)  P\left(  \chi_{r+2v}^{2}>\chi_{r,\alpha_0}^{2}\right)  ,
\end{align}
where
\[
C_{v}\left(  \boldsymbol{t},\boldsymbol{A}\right)  =\frac{\left(
\mathbf{t}^{T}\boldsymbol{A}\mathbf{t}\right)  ^{v}}{v!2^{v}}e^{-\frac{1}
{2}\mathbf{t}^{T}\boldsymbol{A}\mathbf{t}}.
\]

\end{enumerate}
\end{theorem}

\begin{proof}
Let us denote $\boldsymbol{\eta}_{n}^{\ast}=\boldsymbol{T}_{\alpha
}(H_{n,\epsilon,(y_{t},\boldsymbol{x}_{t})}^{P})$. Then, the asymptotic
distribution of the MDPDE $\widehat{\boldsymbol{\eta}}_{\alpha}$ under
$H_{n,\epsilon,(y_{t},\boldsymbol{x}_{t})}^{P}$ yields
\begin{align}
\sqrt{n}\left(  \widehat{\boldsymbol{\eta}}_{\alpha}-\boldsymbol{\eta}%
_{n}^{\ast}\right)  \underset{n\rightarrow\infty}{\overset{L}{\longrightarrow
}}\mathcal{N}\left(  \boldsymbol{0,\Sigma}_{\alpha}\left(  \boldsymbol{\beta
}_{0}\right)  \right) .\label{eq1}%
\end{align}
Now, using a suitable Taylor series approximation and the above asymptotic
distribution, we get
\begin{align}
W_{n}\left(  \widehat{\boldsymbol{\eta}}_{\alpha}\right)   &  =n\boldsymbol{m}%
\left(  \widehat{\boldsymbol{\eta}}_{\alpha}\right)  ^{T}\left[
\boldsymbol{M}^{T}\left(  \boldsymbol{\eta}_{0}\right)  \boldsymbol{\Sigma
}_{\alpha}\left(  \boldsymbol{\eta}_{0}\right)  \boldsymbol{M}\left(
\boldsymbol{\eta}_{0}\right)  \right]  ^{-1}\boldsymbol{m}\left(
\widehat{\boldsymbol{\eta}}_{\alpha}\right) \nonumber\\
&  =n\boldsymbol{m}\left(  \boldsymbol{\eta}_{n}^{\ast}\right)  ^{T}\left[
\boldsymbol{M}^{T}\left(  \boldsymbol{\eta}_{0}\right)  \boldsymbol{\Sigma
}_{\alpha}\left(  \boldsymbol{\eta}_{0}\right)  \boldsymbol{M}\left(
\boldsymbol{\eta}_{0}\right)  \right]  ^{-1}\boldsymbol{m}\left(
\boldsymbol{\eta}_{n}^{\ast}\right) \nonumber\\
&  ~+n\left(  \widehat{\boldsymbol{\eta}}_{\alpha}-\boldsymbol{\eta}_{n}%
^{\ast}\right)  ^{T}\boldsymbol{M}\left(  \boldsymbol{\eta}_{0}\right)
\left[  \boldsymbol{M}^{T}\left(  \boldsymbol{\eta}_{0}\right)
\boldsymbol{\Sigma}_{\alpha}\left(  \boldsymbol{\eta}_{0}\right)
\boldsymbol{M}\left(  \boldsymbol{\eta}_{0}\right)  \right]  ^{-1}%
\boldsymbol{M}^{T}\left(  \boldsymbol{\eta}_{0}\right) \left(
\widehat{\boldsymbol{\eta}}_{\alpha}-\boldsymbol{\eta}_{n}^{\ast}\right)
\nonumber\\
&  ~+n\left(  \widehat{\boldsymbol{\eta}}_{\alpha}-\boldsymbol{\eta}_{n}%
^{\ast}\right)  ^{T}\boldsymbol{M}\left(  \boldsymbol{\eta}_{0}\right)
\left[  \boldsymbol{M}^{T}\left(  \boldsymbol{\eta}_{0}\right)
\boldsymbol{\Sigma}_{\alpha}\left(  \boldsymbol{\eta}_{0}\right)
\boldsymbol{M}\left(  \boldsymbol{\eta}_{0}\right)  \right]  ^{-1}%
\boldsymbol{m}\left(  \boldsymbol{\eta}_{n}^{\ast}\right)  +o_{P}(1).\nonumber
\end{align}
Again, another Taylor series approximation yields
\begin{align}
\sqrt{n}(\boldsymbol{\eta}_{n}^{\ast}-\boldsymbol{\eta}_{0})  &
=\boldsymbol{d}+\epsilon\mathcal{IF}\left(  (y_{t},\boldsymbol{x}%
_{t}),\boldsymbol{T}_{\alpha},H_{\boldsymbol{\eta}_{0}}\right)  +o_{p}%
(\boldsymbol{1}_{p})\nonumber\\
&  =\widetilde{\boldsymbol{d}}_{\epsilon,(y_{t},\boldsymbol{x}_{t}),\alpha
}(\boldsymbol{\theta}_{0})+o_{p}(\boldsymbol{1}_{p}), \label{CDT}%
\end{align}
and hence
\begin{align}
\sqrt{n}\boldsymbol{m}\left(  \boldsymbol{\eta}_{n}^{\ast}\right)   &
=\boldsymbol{M}^{T}\left(  \boldsymbol{\eta}_{0}\right) \sqrt{n}%
(\boldsymbol{\eta}_{n}^{\ast}-\boldsymbol{\eta}_{0})+o_{p}(\boldsymbol{1}%
_{p})\nonumber\\
&  =\boldsymbol{M}^{T}\left(  \boldsymbol{\eta}_{0}\right)
\widetilde{\boldsymbol{d}}_{\epsilon,(y_{t},\boldsymbol{x}_{t}),\alpha
}(\boldsymbol{\eta}_{0})+o_{p}(\boldsymbol{1}_{p}), \label{EQ:eq1}%
\end{align}
using $\boldsymbol{m}\left(  \boldsymbol{\eta}_{0}\right)  =0$. Therefore,
combining all the above results, we get
\[
W_{n}\left(  \widehat{\boldsymbol{\eta}}_{\alpha}\right)  =\boldsymbol{Z}%
_{n}^{T}\left[  \boldsymbol{M}^{T}\left(  \boldsymbol{\eta}_{0}\right)
\boldsymbol{\Sigma}_{\alpha}\left(  \boldsymbol{\eta}_{0}\right)
\boldsymbol{M}\left(  \boldsymbol{\eta}_{0}\right)  \right]  ^{-1}%
\boldsymbol{Z}_{n}+o_{p}(1),
\]
where
\[
\boldsymbol{Z}_{n}=\boldsymbol{M}^{T}\left(  \boldsymbol{\eta}_{0}\right)
\sqrt{n}\left(  \widehat{\boldsymbol{\eta}}_{\alpha}-\boldsymbol{\eta}%
_{n}^{\ast}\right)  +\boldsymbol{M}^{T}\left(  \boldsymbol{\eta}_{0}\right)
\widetilde{\boldsymbol{d}}_{\epsilon,(y_{t},\boldsymbol{x}_{t}),\alpha
}(\boldsymbol{\eta}_{0}).
\]
But, by (\ref{eq1}),
\[
\boldsymbol{Z}_{n}\underset{n\rightarrow\infty}{\overset{\mathcal{L}%
}{\longrightarrow}}\mathcal{N}\left(  \boldsymbol{M}^{T}\left(
\boldsymbol{\eta}_{0}\right) \widetilde{\boldsymbol{d}}_{\epsilon
,(y_{t},\boldsymbol{x}_{t}),\alpha}(\boldsymbol{\eta}_{0}),\ \boldsymbol{M}%
^{T}\left(  \boldsymbol{\eta}_{0}\right) \boldsymbol{\Sigma}_{\alpha}\left(
\boldsymbol{\eta}_{0}\right)  \boldsymbol{M}\left(  \boldsymbol{\eta}%
_{0}\right)  \right)  ,
\]
which implies that $W_{n}\left(  \widehat{\boldsymbol{\eta}}_{\alpha}\right)
\underset{n\rightarrow\infty}{\overset{\mathcal{L}}{\longrightarrow}}\chi
_{r}^{2}(\delta),$ a non-central $\chi^{2}$ random variable with degrees of
freedom $r$ and non-centrality parameter $\delta$ as defined in
(\ref{EQ:delta}).

The second part of the theorem follows by the infinite series expansion of a the
above non-central $\chi^{2}$ distribution in terms of the central chi-square
variables as 
\begin{align}
\pi(\boldsymbol{\eta}_{n},\epsilon,(y_{t},\boldsymbol{x}_{t}))  &
=\lim_{n\rightarrow\infty}P_{H_{n,\epsilon,(y_{t},\boldsymbol{x}_{t})}^{P}%
}(W_{n}>\chi_{r,\alpha_0}^{2}) \cong P(\chi_{r}^{2}(\delta) >\chi_{r,\alpha_0}%
^{2})\nonumber\\
&  =\sum\limits_{v=0}^{\infty} C_{v}\left(  \boldsymbol{M}^{T}\left(
\boldsymbol{\eta}_{0}\right) \widetilde{\boldsymbol{d}}_{\epsilon
,(y_{t},\boldsymbol{x}_{t}),\alpha}(\boldsymbol{\eta}_{0}), \ \left[
\boldsymbol{M}^{T}\left(  \boldsymbol{\eta}_{0}\right)  \boldsymbol{\Sigma
}_{\alpha}\left(  \boldsymbol{\eta}_{0}\right)  \boldsymbol{M}\left(
\boldsymbol{\eta}_{0}\right)  \right]  ^{-1}\right)  P\left(  \chi_{r+2v}%
^{2}>\chi_{r,\alpha_0}^{2}\right)  .\nonumber
\end{align}

\end{proof}

\bigskip
Note that, substituting $\epsilon=0$ in Theorem \ref{THM:7asymp_power_one}, we
get an alternative expression for the asymptotic power function of our
proposed Wald-type test statistics under the contiguous alternatives
$\boldsymbol{\eta}_{n}=\boldsymbol{\eta}_{0}+n^{-1/2}\boldsymbol{d}$ as
\begin{equation}
\pi(\boldsymbol{\eta}_{n}) =\pi(\boldsymbol{\eta}_{n},0,(y_{t},\boldsymbol{x}%
_{t})) \cong\sum\limits_{v=0}^{\infty}C_{v}\left(  \boldsymbol{M}^{T}\left(
\boldsymbol{\eta}_{0}\right) \boldsymbol{d}, \left[  \boldsymbol{M}^{T}\left(
\boldsymbol{\eta}_{0}\right) \boldsymbol{\Sigma}_{\alpha}\left(
\boldsymbol{\eta}_{0}\right)  \boldsymbol{M}\left(  \boldsymbol{\eta}%
_{0}\right)  \right]  ^{-1}\right)  P\left(  \chi_{r+2v}^{2}>\chi_{r,\alpha_0
}^{2}\right)  .\nonumber
\end{equation}

Further, substituting $\boldsymbol{d}=\boldsymbol{0}_{r}$ in Theorem
\ref{THM:7asymp_power_one}, we can derive the asymptotic distribution of the
Wald-type test statistics $W_{n}$ under $H_{n,\epsilon,(y_{t},\boldsymbol{x}%
_{t})}^{L}$ which is non-central chi-square with $r$ degrees of freedom and
non-centrality parameter
\[
\epsilon^{2}\mathcal{IF}((y_{t},\boldsymbol{x}_{t});T_{\alpha}%
,H_{\boldsymbol{\beta}_{0}})^{T} \boldsymbol{M}\left(  \boldsymbol{\eta}%
_{0}\right)  \left[  \boldsymbol{M}^{T}\left(  \boldsymbol{\eta}_{0}\right)
\boldsymbol{\Sigma}_{\alpha}\left(  \boldsymbol{\eta}_{0}\right)
\boldsymbol{M}\left(  \boldsymbol{\eta}_{0}\right)  \right]  ^{-1}
\boldsymbol{M}^{T}\left(  \boldsymbol{\eta}_{0}\right)  \mathcal{IF}%
((y_{t},\boldsymbol{x}_{t});T_{\alpha},H_{\boldsymbol{\beta}_{0}}).
\]
Therefore, the asymptotic level under contiguous contamination $H_{n,\epsilon
,(y_{t},\boldsymbol{x}_{t})}^{L}$ turns out to be
\begin{align}
&  \alpha(\epsilon,(y_{t},\boldsymbol{x}_{t})) =\pi(\boldsymbol{\eta}%
_{0},\epsilon,(y_{t},\boldsymbol{x}_{t}))\nonumber\\
&  \cong\sum\limits_{v=0}^{\infty}C_{v}\left(  \epsilon\boldsymbol{M}%
^{T}\left(  \boldsymbol{\eta}_{0}\right) \mathcal{IF}((y_{t},\boldsymbol{x}%
_{t});T_{\alpha},H_{\boldsymbol{\beta}_{0}}), \left[  \boldsymbol{M}%
^{T}\left(  \boldsymbol{\eta}_{0}\right) \boldsymbol{\Sigma}_{\alpha}\left(
\boldsymbol{\eta}_{0}\right)  \boldsymbol{M}\left(  \boldsymbol{\eta}%
_{0}\right)  \right]  ^{-1}\right)  P\left(  \chi_{r+2v}^{2}>\chi_{r,\alpha_0
}^{2}\right)  .\nonumber
\end{align}
Note that, as $\epsilon\rightarrow0$, $\alpha(\epsilon,(y_{t},\boldsymbol{x}%
_{t})) \rightarrow\alpha_{0}$, the nominal level of the test.

Using the above expressions for asymptotic power and level under contiguous
contamination, one can easily derive the PIF and LIF of the proposed Wald-type
test statistics as described in the following theorem.

\begin{theorem}
\label{Theorem10} Assume the conditions of Theorem \ref{THM:7asymp_power_one}
hold. Then, the power and level influence functions of our proposed
Wald-type tests based on $W_{n}$ is given by
\begin{equation}
\mathcal{PIF}((y_{t},\boldsymbol{x}_{t}),W_{n},H_{\boldsymbol{\beta}_{0}})
\cong K_{r}^{\ast}\left(  \boldsymbol{P}\mathbf{d}\right)  ~ \boldsymbol{P}%
\cdot\mathcal{IF}((y_{t},\boldsymbol{x}_{t}),\boldsymbol{T}_{\alpha
},H_{\boldsymbol{\beta}_{0}}), \label{EQ:7PIF_simpleTest1}%
\end{equation}
with $\boldsymbol{P} = \boldsymbol{d}^{T}\boldsymbol{M}\left( \boldsymbol{\eta
}_{0}\right)  \left[ \boldsymbol{M}^{T}\left( \boldsymbol{\eta}_{0}\right)
\boldsymbol{\Sigma}_{\alpha}\left( \boldsymbol{\eta}_{0}\right)
\boldsymbol{M}\left(  \boldsymbol{\eta}_{0}\right)  \right] ^{-1}%
\boldsymbol{M}^{T}\left(  \boldsymbol{\eta}_{0}\right) $ and
\[
K_{r}^{\ast}(s)=e^{-\frac{s}{2}}\sum\limits_{v=0}^{\infty}\frac{s^{v-1}
}{v!2^{v}}\left(  2v-s\right)  P\left(  \chi_{r+2v}^{2}>\chi_{r,\alpha_0}
^{2}\right)  ,
\]
and
\[
\mathcal{LIF}((y_{t},\boldsymbol{x}_{t}),W_{n},H_{\boldsymbol{\beta}_{0}})=0.
\]
Also, the level influence function of any higher order is also identically zero.
\end{theorem}


\begin{proof}
The proof follows by differentiating the expression of $\pi(\boldsymbol{\eta
}_{n},\epsilon,(y_{t},\boldsymbol{x}_{t}))$ from Theorem
\ref{THM:7asymp_power_one} with respect to $\epsilon$ using the chain rule and
is similar to that of Theorem 8 of Ghosh et al.~(2016).

\end{proof}

Note that the above theorem implies the stability of the asymptotic level of
our proposed Wald-type tests with respect to the infinitesimal contamination for any
$\alpha\geq0$. On the other hand the power influence function is bounded
implying the stability of the asymptotic contiguous power only when the
influence function of the MDPDE is bounded, i.e., for $\alpha>0$. The PIF of
the classical Wald-type test based on MLE (at $\alpha=0$) is unbounded
indicating its well-known non-robust nature.

\section{Application: Poisson Regression Model under Random Design\label{sec5}%
}

Poisson regression is a very popular member of the class of GLMs where the underlying
distribution, given by the density $f(y,\boldsymbol{x}^{T}\boldsymbol{\beta},\phi)=f_{P}%
(y,\boldsymbol{x}^{T}\boldsymbol{\beta})$, is  Poisson with mean
$E(Y|\boldsymbol{x}) = e^{\boldsymbol{x}^{T}\boldsymbol{\beta}}$,
so that
\begin{align}
f_{P}(y,\boldsymbol{x}^{T}\boldsymbol{\beta}) = \frac{e^{y(\boldsymbol{x}%
^{T}\boldsymbol{\beta})}}{y!}e^{-e^{\boldsymbol{x}^{T}\boldsymbol{\beta}}},
~~~y=0,1,2, \ldots.\nonumber
\end{align}
Hence, in terms of the general model density given in Equation (\ref{1}), we have $\phi=1$ and $\boldsymbol{\theta
}= \boldsymbol{x}_{i}^{T}\boldsymbol{\beta}$, $b(\boldsymbol{\theta}) =
e^{\boldsymbol{\theta}}$, 
and the link function $g$ is the natural logarithm function.
Also, note that $V(Y|\boldsymbol{x}) = E(Y|\boldsymbol{x}) = e^{\boldsymbol{x}%
^{T}\boldsymbol{\beta}}$. Additionally, we assume that the covariates
$\boldsymbol{X}$ are random having distribution function $G(\boldsymbol{x})$, which is
generally normal for continuous covariates. This regression model is widely
used in practice for modeling count data like total number of occurrences of a
particular disease in medical sciences, number of failures in reliability or
survival analysis, etc.

Note that, as $\phi=1$ known for the case of Poisson regression the parameter
of interest is $\boldsymbol{\eta}=\boldsymbol{\beta}$. The MDPDE of
$\boldsymbol{\beta}$ can then be obtained by solving only one (unbiased)
estimating equation (\ref{3}) which has the simplified form for Poisson
regression as
\begin{equation}
\sum_{i=1}^{n} \left[  \gamma_{1,\alpha}(\boldsymbol{x}_{i}) -\left(  y_{i}-
e^{\boldsymbol{x}_{i}^{T}\boldsymbol{\beta}}\right)  f_{P}^{\alpha}%
(y_{i},\boldsymbol{x}_{i}^{T}\boldsymbol{\beta})\right]  \boldsymbol{x}_{i}
=\boldsymbol{0}, \label{EQ:est_eqn_PoissReg}%
\end{equation}
where $\gamma_{1,\alpha}(\boldsymbol{x}) = \displaystyle\sum_{y=0}^{\infty} (y
- e^{\boldsymbol{x}^{T}\boldsymbol{\beta}}) f_{P}^{1+\alpha}(y,\boldsymbol{x}%
^{T}\boldsymbol{\beta}).$ For the particular case of $\alpha=0$, we have
$\gamma_{1,0}(\boldsymbol{x}) =0$ and hence this estimating equation further
simplifies to
\begin{equation}
\sum_{i=1}^{n} \left(  y_{i}- e^{\boldsymbol{x}_{i}^{T}\boldsymbol{\beta}%
}\right)  \boldsymbol{x}_{i} = \boldsymbol{0},
\end{equation}
which is nothing but the likelihood score equation of the  maximum likelihood estimator (MLE)
of $\boldsymbol{\beta}$.

Now, the asymptotic distribution of the MDPDE $\widehat{\boldsymbol{\beta}%
}_{\alpha}$ of $\boldsymbol{\beta}$ can be derived directly from the results of  Section
\ref{sec2.1}. In particular, under the model distribution with true parameter
value $\boldsymbol{\beta}_{0}$, we have 
\[
\sqrt{n}(\widehat{\boldsymbol{\beta}}_{\alpha}-\boldsymbol{\beta}_{0})
\underset{n\rightarrow\infty}{\longrightarrow} \mathcal{N}(\boldsymbol{0}%
,\boldsymbol{J}_{\alpha}^{-1}(\boldsymbol{\beta}_{0})\boldsymbol{K}_{\alpha
}(\boldsymbol{\beta}_{0}) \boldsymbol{J}_{\alpha}^{-1}(\boldsymbol{\beta}%
_{0})),
\]
where we now have $\boldsymbol{K}_{\alpha}(\boldsymbol{\beta}) ={\int%
_{\mathcal{X}}}\left(  \gamma_{11,2\alpha}(\boldsymbol{x})-\gamma_{1,\alpha
}^{2}(\boldsymbol{x})\right)  \boldsymbol{x}\boldsymbol{x}^{T}%
dG(\boldsymbol{x})$ and $\boldsymbol{J}_{\alpha}(\boldsymbol{\beta})
={\int_{\mathcal{X}}}\gamma_{11,\alpha}(\boldsymbol{x})\boldsymbol{x}%
\boldsymbol{x}^{T}dG(\boldsymbol{x}) $
with $\gamma_{11,\alpha}(\boldsymbol{x}) = \displaystyle\sum_{y=0}^{\infty} (y
- e^{\boldsymbol{x}^{T}\boldsymbol{\beta}})^{2} f_{P}^{1+\alpha}%
(y,\boldsymbol{x}^{T}\boldsymbol{\beta}).$ At $\alpha=0$, one can show that
$\gamma_{11,\alpha}(\boldsymbol{x})=e^{\boldsymbol{x}^{T}\boldsymbol{\beta}}$
and hence $\boldsymbol{K}_{\alpha}(\boldsymbol{\beta})= \boldsymbol{J}%
_{\alpha}(\boldsymbol{\beta}) ={\int_{\mathcal{X}}}e^{\boldsymbol{x}%
^{T}\boldsymbol{\beta}}\boldsymbol{x}\boldsymbol{x}^{T}dG(\boldsymbol{x})$,
which is exactly the Fisher information matrix under the present set-up generating
the asymptotic distribution of the MLE $\widehat{\boldsymbol{\beta}}_{0}$.
Based on these asymptotic distributions, one can compute the asymptotic
relative efficiencies of our MDPDEs over $\alpha$ which are presented in Table
\ref{TAB:ARE} for the case of a scalar ($k=1$) normally distributed covariate
$\boldsymbol{x}$. Clearly, as expected from the literature of the
MDPDE in any other model, the ARE decreases slightly as $\alpha$ increases but
this loss in efficiency is not substantial at small positive $\alpha$. And,
with this small price in asymptotic efficiency, we gain high robustness
properties of our MDPDEs with $\alpha>0$.

\begin{table}[h]
\caption{Asymptotic relative efficiency of MDPDEs of $\boldsymbol{\beta}$ over
$\alpha$ under a Poisson regression model with a scalar ($k=1$) covariate
$\boldsymbol{x}\sim N(\mu_{x}, 1)$  and different true parameter values
$\boldsymbol{\beta}_{0}$}%
\label{TAB:ARE}%
\centering
\begin{tabular}
[c]{|l|l|cccccccc|}\hline
$\mu_{x}$ & $\boldsymbol{\beta}_{0}$ & \multicolumn{8}{c|}{$\alpha$}\\
&  & 0 & 0.05 & 0.1 & 0.25 & 0.4 & 0.5 & 0.7 & 1\\\hline\hline
0 & 1 & 1.000 & 0.995 & 0.985 & 0.927 & 0.849 & 0.793 & 0.671 & 0.489\\
0 & 0.5 & 1.000 & 0.996 & 0.985 & 0.931 & 0.861 & 0.811 & 0.713 & 0.576\\
1 & 1 & 1.000 & 0.995 & 0.986 & 0.932 & 0.859 & 0.807 & 0.701 & 0.550\\
1 & 0.5 & 1.000 & 0.997 & 0.988 & 0.940 & 0.880 & 0.839 & 0.757 & 0.646\\
5 & 1 & 1.000 & 0.996 & 0.986 & 0.927 & 0.848 & 0.791 & 0.676 & 0.516\\
5 & 0.5 & 1.000 & 0.996 & 0.987 & 0.937 & 0.872 & 0.826 & 0.736 &
0.615\\\hline
\end{tabular}
\end{table}

To see such robustness advantages of our MDPDEs $\widehat{\boldsymbol{\beta}%
}_{\alpha}$, we consider the influence function of the MDPDE functional
${T}_{\alpha}^{\beta}$ of $\boldsymbol{\beta}$ from Section
\ref{sec2.2}. This influence function can be simplified for the present case
of Poisson regression model at the model distribution with parameter value
$\boldsymbol{\beta}$ as
\begin{align}
&  IF((y_{t}, \boldsymbol{x}_{t}), {T}_{\alpha}^{\beta}, H_{0}) =
\left(  {\int_{\mathcal{X}}}\gamma_{11,\alpha}(\boldsymbol{x})\boldsymbol{x}%
\boldsymbol{x}^{T}dG(\boldsymbol{x})\right)  ^{-1} \boldsymbol{x}_{t} \left[
\frac{(y_{t} - e^{\boldsymbol{x}_{t}^{T}\boldsymbol{\beta}})}{(y_{t}%
!)^{\alpha}} e^{\alpha\left[  y_{t}(\boldsymbol{x}_{t}^{T}\boldsymbol{\beta})
- e^{\boldsymbol{x}_{t}^{T}\boldsymbol{\beta}}\right]  } - \gamma_{1,\alpha
}(\boldsymbol{x}_{t})\right]  .\nonumber
\end{align}

Note that the above influence function is bounded at $\alpha>0$ and unbounded
at $\alpha=0$. This implies the robustness of the MDPDEs with $\alpha>0$ and
the non-robust nature of the MLE at $\alpha=0$. In particular, the influence
function of the MLE under the Poisson regression model is a straight line
(unbounded in both outliers in response, $y_{t}$, and leverage points in
covariate space, $\boldsymbol{x}_{t}$) and is given by
\[
IF((y_{t}, \boldsymbol{x}_{t}), {T}_{0}^{\beta}, H_{0}) = \left(
{\int_{\mathcal{X}}}e^{\boldsymbol{x}^{T}\boldsymbol{\beta}}\boldsymbol{x}%
\boldsymbol{x}^{T}dG(\boldsymbol{x})\right) ^{-1} \boldsymbol{x}_{t} (y_{t} -
e^{\boldsymbol{x}_{t}^{T}\boldsymbol{\beta}}).
\]
Figure \ref{FIG:IF_MDPDE} presents these influence functions for different $\alpha$, 
when $\boldsymbol{x}$ is a scalar ($k=1$) continuous variable 
having a normal distribution. Note that the influence function of
the classical Wald test at $\alpha=0$ is unbounded for $y_{t}%
\rightarrow\infty$ for any fixed $\boldsymbol{x}_{t}$ (outlier in response) as
well as for $\boldsymbol{x}_{t} \rightarrow\infty$ with small $y_{t}$ or
$\boldsymbol{x}_{t} \rightarrow-\infty$ with larger $y_{t}$ (leverage points).
On the contrary, influence functions of the MDPDEs with $\alpha>0$ are bounded
in both the cases indicating their robustness against outliers in both $y$ and
$x$-spaces. Also, the nature of the influence function (and hence robustness
of the corresponding estimators) remains invariant with respect any change in
the covariate mean $\mu_{x}$ (only the magnitude of the influence function
changes). Further, the supremum of the IF in absolute value decreases as
$\alpha$ increases, indicating the increasing robustness of the MDPDEs with
increasing $\alpha$.

\begin{figure}[!h]
\centering
\subfloat[$\mu_x=0,~\alpha=0$]{
		\includegraphics[width=0.4\textwidth]{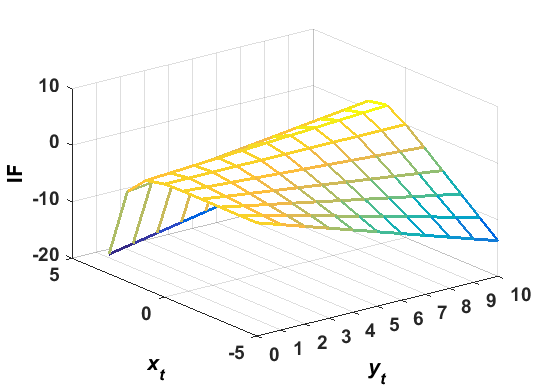}
		\label{FIG:1aaa}} ~
\subfloat[$\mu_x=0,~\alpha=0.1$]{
		\includegraphics[width=0.4\textwidth]{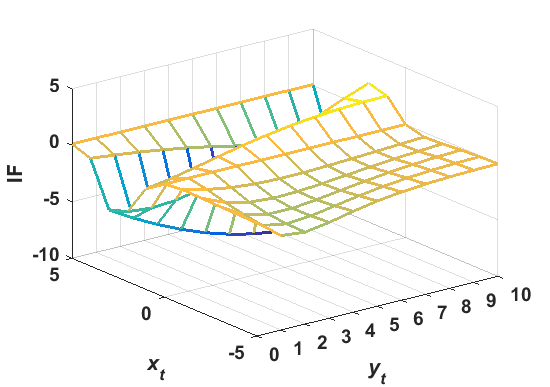}
		\label{FIG:1ba}} ~
\newline
\subfloat[$\mu_x=0,~\alpha=0.25$]{
		\includegraphics[width=0.4\textwidth]{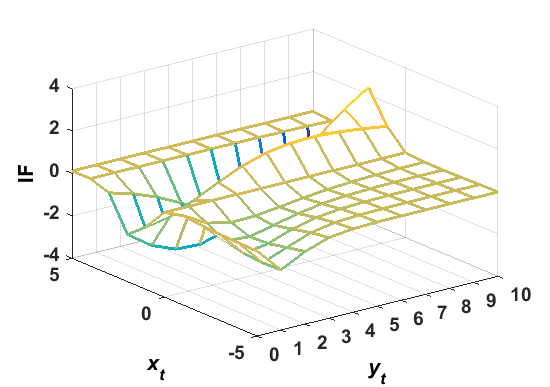}
		\label{FIG:1ca}} ~
\subfloat[$\mu_x=0,~\alpha=0.5$]{
		\includegraphics[width=0.4\textwidth]{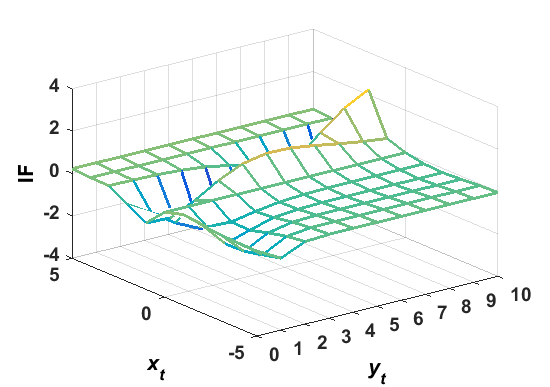}
		\label{FIG:1da}} ~
\caption{Influence function of MDPDE $\boldsymbol{T}_{\alpha}^{\beta}$ for a
Poisson regression model with $k=1$, true parameter $\boldsymbol{\beta}=1$ and
covariate $\boldsymbol{x}\sim N(\mu_{x},1)$.
The figures for other values of $\mu_x$ are similar and hence not reported for brevity}
\label{FIG:IF_MDPDE}%
\end{figure}

Now, consider the problem of testing the general linear hypothesis of
$\boldsymbol{\beta}$ under the Poisson regression model, i.e., consider the
hypothesis
\begin{equation}
H_{0}:\boldsymbol{L}\boldsymbol{\beta}=\boldsymbol{l}_{0} \text{ versus }%
H_{1}:\boldsymbol{L}\boldsymbol{\beta}\neq\boldsymbol{l}_{0}, \label{3.10}%
\end{equation}
where $\boldsymbol{L}$ is a full rank matrix of order $r\times k$, with $r<k$
($\mathrm{rank}\left(  \boldsymbol{L}\right)  =r$), and $\boldsymbol{l}_{0}$
is an $r$-dimensional vector, both of known values. We assume that
$\mathrm{rank}\left(  \boldsymbol{L},\boldsymbol{l}_{0}\right)  =r$. This
clearly belongs to the general class of hypothesis considered in
(\ref{EQ:hyp}) with $\boldsymbol{m}(\boldsymbol{\eta}) = \boldsymbol{m}%
(\boldsymbol{\beta}) = \boldsymbol{L}\boldsymbol{\beta}-\boldsymbol{l}_{0}$
and $\boldsymbol{M}(\boldsymbol{\beta}) = \boldsymbol{L}^{T}$ (since $\phi=1$
is known here). Then the proposed MDPDE based Wald-type test statistics for
testing (\ref{3.10}) is given by
\begin{equation}
W_{n}(\widehat{\boldsymbol{\beta}}_{\alpha}) =n(\boldsymbol{L}%
\widehat{\boldsymbol{\beta}}_{\alpha}-\boldsymbol{l}_{0})^{T} \left[
\boldsymbol{L}\boldsymbol{\Sigma}_{\alpha}(\widehat{\boldsymbol{\beta}%
}_{\alpha})\boldsymbol{L}^{T}\right] ^{-1} (\boldsymbol{L}%
\widehat{\boldsymbol{\beta}}_{\alpha}-\boldsymbol{l}_{0}). \label{3.11}%
\end{equation}
By Theorem \ref{THM:asymp_null}, under $H_{0}$, the above Wald-type test statistics
$W_{n}(\widehat{\boldsymbol{\beta}}_{\alpha})$ asymptotically follow a 
$\chi_{r}^{2}$ distribution. The tests are also consistent at any fixed
alternative from Corollary \ref{CORR:consistency}. We will now derive their 
asymptotic power under contiguous alternatives $H_{1,n}: \boldsymbol{\beta
}_{n} = \boldsymbol{\beta}_{0} + n^{-1/2}\boldsymbol{d}$, where
$\boldsymbol{\beta}_{0}$ is the true null parameter value satisfying
$\boldsymbol{L}\boldsymbol{\beta}_{0}=\boldsymbol{l}_{0}$. From Theorem
\ref{THM:Asump_contg}, we get the asymptotic distribution of our Wald-type test
statistic $W_{n}(\widehat{\boldsymbol{\beta}}_{\alpha})$ to be a non-central
chi-square distribution with degrees of freedom $r$ and non-centrality
parameter $\boldsymbol{d}^{T}\boldsymbol{L}^{T} \left[ \boldsymbol{L}%
\boldsymbol{\Sigma}_{\alpha}(\boldsymbol{\eta}_{0})\boldsymbol{L}^{T}\right]
^{-1} \boldsymbol{L}\boldsymbol{d}$. Hence the asymptotic contiguous power can
be obtained from the distribution function of this non-central chi-square
distribution, which is presented in Table \ref{TAB:ContPower} for the case
$k=1$ with a normally distributed covariate $\boldsymbol{x}$. One can clearly
observe that the asymptotic contiguous power for any fixed $\boldsymbol{d}$
decreases slightly with increasing $\alpha$, but the loss in power in not
quite significant. Notice the similarity with the nature of ARE of the
corresponding MDPDE $\widehat{\boldsymbol{\beta}}_{\alpha}$ from Table
\ref{TAB:ARE}, because the asymptotic contiguous power is directly related to
the asymptotic variance (and hence to the asymptotic efficiency) of the
estimator used.

\begin{table}[h]
\caption{Asymptotic power of the proposed Wald-type test statistics in
(\ref{3.11}) for testing (\ref{3.10}) under the contiguous alternatives
$H_{1,n}$ for a Poisson regression model with a scalar ($k=1$) covariate
$\boldsymbol{x}\sim N(\mu_{x}, 1)$ and different null parameter values
$\boldsymbol{\beta}_{0}$}%
\label{TAB:ContPower}
\centering
\begin{tabular}
[c]{|l|l|l|cccccccc|}\hline
$\boldsymbol{d}$ & $\mu_{x}$ & $\boldsymbol{\beta}_{0}$ &
\multicolumn{8}{c|}{$\alpha$}\\
&  &  & 0 & 0.05 & 0.1 & 0.25 & 0.4 & 0.5 & 0.7 & 1\\\hline\hline
1 & 0 & 1 & 0.445 & 0.443 & 0.440 & 0.418 & 0.389 & 0.368 & 0.320 & 0.247\\
1 & 0 & 0.5 & 0.236 & 0.235 & 0.233 & 0.222 & 0.209 & 0.200 & 0.181 & 0.156\\
1 & 1 & 1 & 0.998 & 0.998 & 0.997 & 0.996 & 0.993 & 0.990 & 0.979 & 0.943\\
1 & 1 & 0.5 & 0.669 & 0.667 & 0.663 & 0.642 & 0.613 & 0.593 & 0.550 & 0.486\\
1 & 5 & 1 & 1.000 & 1.000 & 1.000 & 1.000 & 1.000 & 1.000 & 1.000 & 1.000\\
1 & 5 & 0.5 & 1.000 & 1.000 & 1.000 & 1.000 & 1.000 & 1.000 & 1.000 &
1.000\\\hline
2 & 0 & 1 & 0.954 & 0.953 & 0.951 & 0.939 & 0.919 & 0.900 & 0.847 & 0.721\\
2 & 0 & 0.5 & 0.696 & 0.695 & 0.690 & 0.665 & 0.632 & 0.606 & 0.551 & 0.467\\
2 & 1 & 1 & 1.000 & 1.000 & 1.000 & 1.000 & 1.000 & 1.000 & 1.000 & 1.000\\
2 & 1 & 0.5 & 0.998 & 0.998 & 0.997 & 0.996 & 0.994 & 0.992 & 0.986 &
0.971\\\hline
\end{tabular}
\end{table}

As in the case of  the MDPDE, we indeed gain high robustness of the proposed
Wald-type test statistics with $\alpha>0$ at a small cost  in
asymptotic contiguous power. To see this, we consider the influence function
analysis for the Poisson regression model following the general theory
developed in Section \ref{sec4}. In particular, the first order influence
function of the Wald-type test statistics is always zero and corresponding
second order influence function for testing (\ref{3.10}) under the Poisson
regression model at the null distribution $H_{0}$ with true parameter value
$\boldsymbol{\beta}_{0}$ simplifies to
\begin{align}
&  \mathcal{IF}_{2}((y_{t},\boldsymbol{x}_{t}),W_{\alpha},H_{0})\nonumber\\
& = \left\{ \boldsymbol{x}_{t}^{T}\boldsymbol{J}_{\alpha}(\boldsymbol{\beta
}_{0})^{-1}\boldsymbol{L}^{T} \left[ \boldsymbol{L}\boldsymbol{\Sigma}%
_{\alpha}(\boldsymbol{\beta}_{0})\boldsymbol{L}^{T}\right] ^{-1}
\boldsymbol{L}\boldsymbol{J}_{\alpha}(\boldsymbol{\beta}_{0})^{-1}%
\boldsymbol{x}_{t}\right\}  \left[  \frac{(y_{t} - e^{\boldsymbol{x}_{t}%
^{T}\boldsymbol{\beta}_{0}})}{(y_{t}!)^{\alpha}} e^{\alpha\left[
y_{t}(\boldsymbol{x}_{t}^{T}\boldsymbol{\beta}_{0}) - e^{\boldsymbol{x}%
_{t}^{T}\boldsymbol{\beta}_{0}}\right] } - \gamma_{1,\alpha}^{(0)}%
(\boldsymbol{x}_{t})\right] ^{2},\nonumber
\end{align}
with $\gamma_{1,\alpha}^{(0)}(\boldsymbol{x}) = \displaystyle\sum
_{y=0}^{\infty} (y - e^{\boldsymbol{x}^{T}\boldsymbol{\beta}_{0}})
f_{P}^{1+\alpha}(y,\boldsymbol{x}^{T}\boldsymbol{\beta}_{0}).$ Similarly,
while considering the level and power robustness of the proposed Wald-type
test statistics for testing (\ref{3.10}) under Poisson regression model, the
LIF is always zero from Theorem \ref{Theorem10} and the PIF at the null
distribution $H_{0}$ simplifies to
\begin{align}
\mathcal{PIF}((y_{t},\boldsymbol{x}_{t}),W_{n},H_{0})\cong K_{r}^{\ast}\left(
\boldsymbol{P}\mathbf{d}\right)  \boldsymbol{P}\boldsymbol{J}_{\alpha
}(\boldsymbol{\beta}_{0})^{-1}\boldsymbol{x}_{t} \left[  \frac{(y_{t} -
e^{\boldsymbol{x}_{t}^{T}\boldsymbol{\beta}_{0}})}{(y_{t}!)^{\alpha}}
e^{\alpha\left[  y_{t}(\boldsymbol{x}_{t}^{T}\boldsymbol{\beta}_{0}) -
e^{\boldsymbol{x}_{t}^{T}\boldsymbol{\beta}_{0}}\right] } - \gamma_{1,\alpha
}^{(0)}(\boldsymbol{x}_{t})\right],
\end{align}
where now we have $\boldsymbol{P}= \boldsymbol{d}^{T}\boldsymbol{L}^{T} \left[
\boldsymbol{L}\boldsymbol{\Sigma}_{\alpha}(\boldsymbol{\beta}_{0}%
)\boldsymbol{L}^{T}\right] ^{-1}\boldsymbol{L}$ and $K_{r}^{\ast}(\cdot)$ is
as defined in Theorem \ref{Theorem10}. Note that, both the second order
influence function of the Wald-type test statistics and its power influence
function are bounded for $\alpha>0$ implying robustness of our proposal. On
the other hand, both are unbounded at $\alpha=0$ demonstrating the well-known
non-robust nature of the classical Wald test. Figures \ref{FIG:IF2_test} and
\ref{FIG:PIF_test}, respectively, present these influence functions for the
Poisson regression case with $k=1$ and a normally distributed covariate. Note
that these influence functions are, respectively, a quadratic and a
linear function of the corresponding influence function of the MDPDE (illustrated in Figure \ref{FIG:IF_MDPDE}) used in
constructing the Wald-type test statistics and demonstrate (appropriately transformed) bounded behavior.   In particular, their redescending nature with respect to increasing
$\alpha$ is clearly seen from the figures which implies that the robustness of
our proposed Wald-type test statistics increases with increasing $\alpha>0$.

\begin{figure}[h]
\centering
\subfloat[$\alpha=0$]{
		\includegraphics[width=0.4\textwidth]{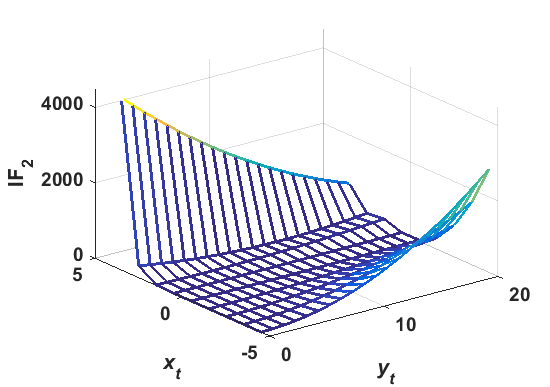}
		\label{FIG:2aa}} ~
\subfloat[$\alpha=0.1$]{
		\includegraphics[width=0.4\textwidth]{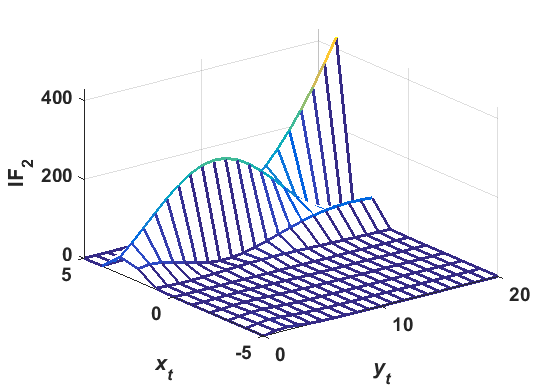}
		\label{FIG:2ab}}\newline\subfloat[$\alpha=0.25$]{
		\includegraphics[width=0.4\textwidth]{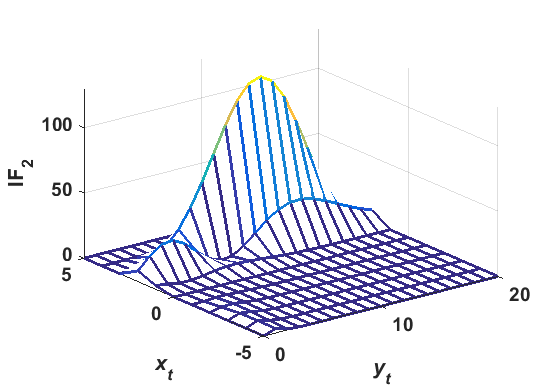}
		\label{FIG:2ba}} ~
\subfloat[$\alpha=0.5$]{
		\includegraphics[width=0.4\textwidth]{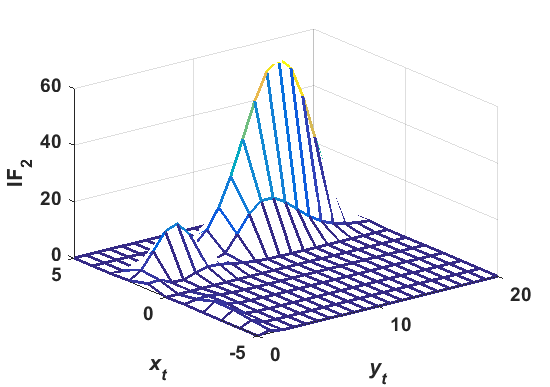}
		\label{FIG:2bb}}
\caption{Second order Influence function of the Wald-type test statistics for
testing (\ref{3.10}) for a Poisson regression model with $k=1$,
$\boldsymbol{L}=1$, true null parameter $\boldsymbol{\beta}_{0}=1$ and
standard normal covariate $\boldsymbol{x}$}%
\label{FIG:IF2_test}%
\end{figure}

\begin{figure}[h]
\centering
\subfloat[$\alpha=0$]{
		\includegraphics[width=0.4\textwidth]{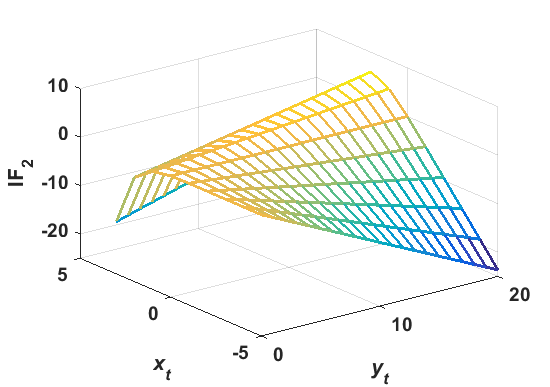}
		\label{FIG:3aa}} ~
\subfloat[$\alpha=0.1$]{
		\includegraphics[width=0.4\textwidth]{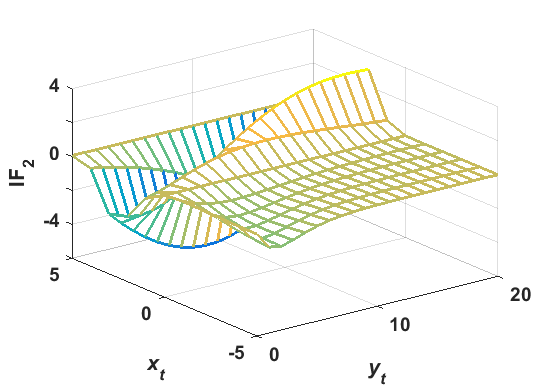}
		\label{FIG:3ab}}\newline\subfloat[$\alpha=0.25$]{
		\includegraphics[width=0.4\textwidth]{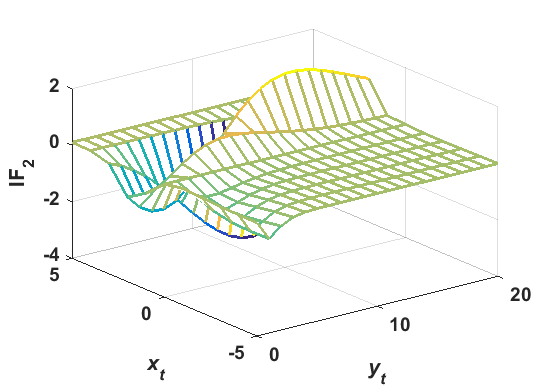}
		\label{FIG:3ba}} ~
\subfloat[$\alpha=0.5$]{
		\includegraphics[width=0.4\textwidth]{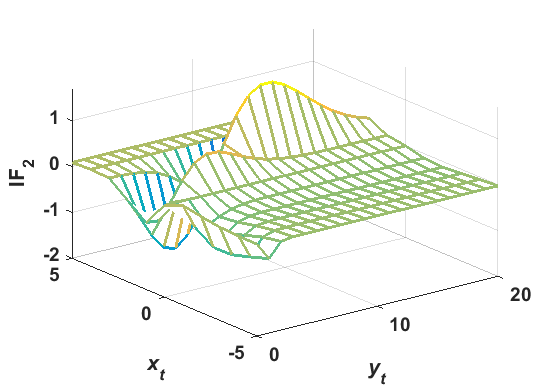}
		\label{FIG:3bb}}
\caption{Power Influence function of the Wald-type test statistics for testing
(\ref{3.10}) for a Poisson regression model with $k=1$, $\boldsymbol{L}=1$,
true null parameter $\boldsymbol{\beta}_{0}=1$, $\boldsymbol{d}=1$ and
standard normal covariate $\boldsymbol{x}$}%
\label{FIG:PIF_test}%
\end{figure}

\section{Simulation Study}
\label{sec6}

In this section, we will present some numerical illustrations for the finite
sample performance of our proposed Wald-type tests under the Poisson
regression model of the previous section through appropriate simulation results.
We start with empirical demonstration of their  robustness properties. 
We consider three explanatory variables  in this
study, so $\boldsymbol{X} = (X_{0}, X_{1}, X_{2}, X_3)^T$, where $X_{0}$ is a vector with all elements equal to one. 
The other three components of $\boldsymbol{X}$ are independently generated from the standard
normal distribution.
The response variable $Y$ is simulated from the Poisson distribution with mean parameter $\exp(\boldsymbol{X}^T \boldsymbol{\beta}_0)$.
The true value of the parameter is taken as
$\boldsymbol{\beta}_{0}=(1,0,0,0)^{T}$. We consider the null hypothesis as 
$H_{0}:(\beta_{1},\beta_{2},\beta_{3})^{T}=(0,0,0)^{T}$. Let us define $\boldsymbol{l_0}=( 0, 0, 0)^{T}$
and
\[
\boldsymbol{L}=\left(
\begin{array}
[c]{ccc}%
0 & 0 & 0\\
1 & 0 & 0\\
0 & 1 & 0\\
0 & 0 & 1
\end{array}
\right)  .
\]
Then the null hypothesis can be written as
$H_{0}: \boldsymbol{L}^{T} \boldsymbol{\beta} = \boldsymbol{l}_0$. 
According to the set up of the simulation the null hypothesis is true. 
So, at first, our interest is to check whether or not the observed levels of different Wald-type tests match with the nominal level at $\alpha_0 = 0.05$. 
The total number of replications is taken as 2000 in this study. 
Here the observed level is measured as the proportion of test statistics exceeding the corresponding chi-square critical value in 2000  replications. 
The results are given in Figure \ref{fig:simulation}(a) where the sample size $n$ varies from 20 to 200. We have used several Wald-type test
statistics, corresponding to different MDPDEs. 
The values of the DPD tuning parameter are taken to be $\alpha=0,\ 0.1,\ 0.2$ and $0.3$, and the  Wald-type test corresponding to $\alpha$ is denoted by $DPD(\alpha)$. As it is previously mentioned, $\alpha=0$ is the classical Wald test
for the Poisson regression model which uses the MLE. The horizontal line in the figure
represents the nominal level of 0.05. It is noticed that all  tests produce almost identical results -- they are slightly
 liberal for small sample sizes and lead to somewhat inflated
observed levels.  However, this discrepancy decreases rapidly as the sample size increases.

In the next simulation study we evaluate the stability of the level of the tests under contamination. So, we
repeated the tests for the same null hypothesis by adding 5\% outliers in the
data. For the outlying observations  the values of the response variable  were altered to $y=15$. Figure \ref{fig:simulation}(b) shows that the level
of the classical Wald test completely breaks down, whereas
Wald-type tests with $\alpha=0.2$ and $\alpha=0.3$ present 
stable levels. The performance of the Wald-type test with $\alpha=0.1$, though much more stable than the classical Wald test, is  relatively poor.

To investigate the power of the Wald-type tests we took the same null hypothesis, 
but changed the true data generating parameter to $\boldsymbol{\beta}^* = \boldsymbol{\beta}_0 - c \boldsymbol{1}_4$, where 
$c = 0.15$ and $\boldsymbol{1}_4$ is a unit vector of length 4. 
The rest of the set up as well as values of $\boldsymbol{L}$ and $\boldsymbol{l}_0$ remained unchanged from the first experiment. The
empirical power functions are calculated in the same manner as the levels of
the tests and plotted in Figure \ref{fig:simulation}(c). Here the classical Wald test is the
most powerful under pure data. However, the performances of other Wald-type tests are also practically as  good as the classical Wald test. Therefore, from Figures \ref{fig:simulation}(a) and (c) we notice that there is hardly any difference among these tests in pure data in terms of the level and power.

Finally, we calculated the power functions of the above
hypothesis  under contaminated data. The true data generating parameter is taken as $\boldsymbol{\beta}^* = \boldsymbol{\beta}_0 - c \boldsymbol{1}_4$, where 
$c = 0.15$ and 5\% of the data are contaminated with $y=15$. The observed
powers of the Wald-type tests are given in Figure \ref{fig:simulation}(d). All
Wald-type test statistics  show stable powers
under contamination, and those powers are almost unchanged as observed in Figure \ref{fig:simulation}(c). On the other hand, the classical Wald test exhibits a drastic loss in power. Notice that the
observed level of the classical Wald test is already very high (around 0.45) at contaminated data, so it is expected to produce a large power just because of the inflated level. But due to their outlier stability, the power of the classical Wald test does not increase with the sample size  at the same rate as the other robust tests. In fact, it shows a relatively significant drop over most of the range considered in our study.   On the whole, the proposed Wald-type test statistics corresponding to
moderately large $\alpha$ appear to be quite competitive to the classical
Wald test for pure data, but they are far better in terms of robustness
properties under contaminated data.

\begin{figure}[h]  \centering
\begin{tabular}{cc}
\vspace{-.9cm} (a) & (b)\\
\includegraphics[height=6cm, width=7cm]{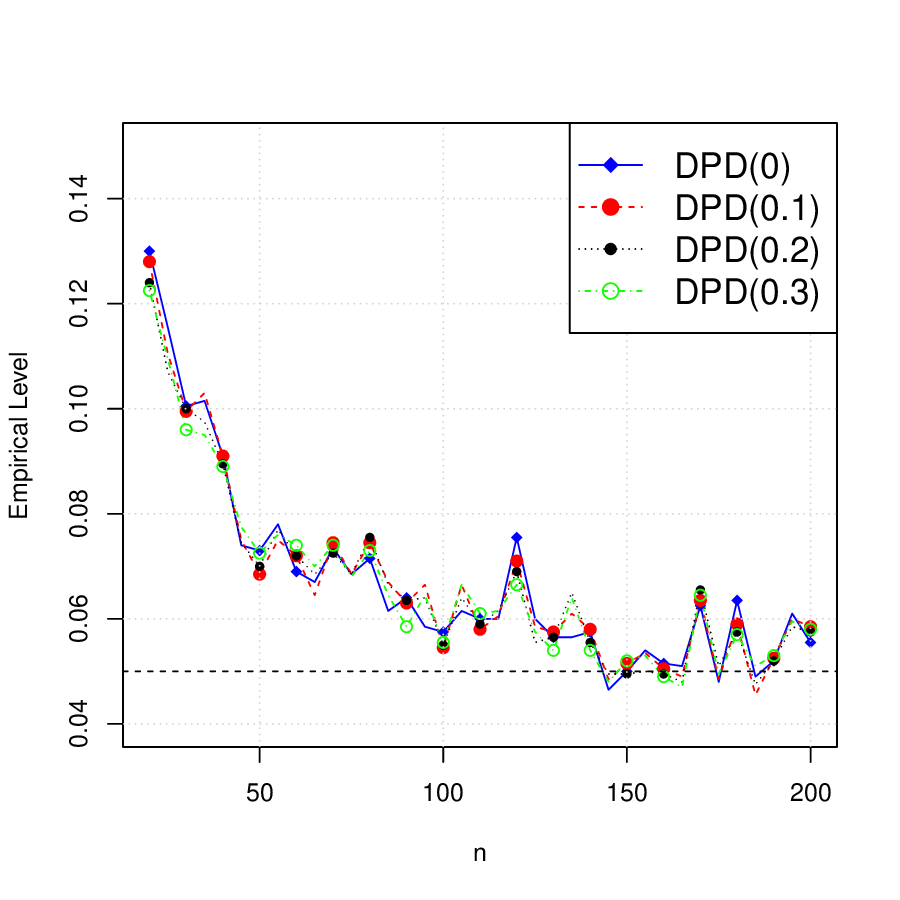} &
 \includegraphics[height=6cm, width=7cm]{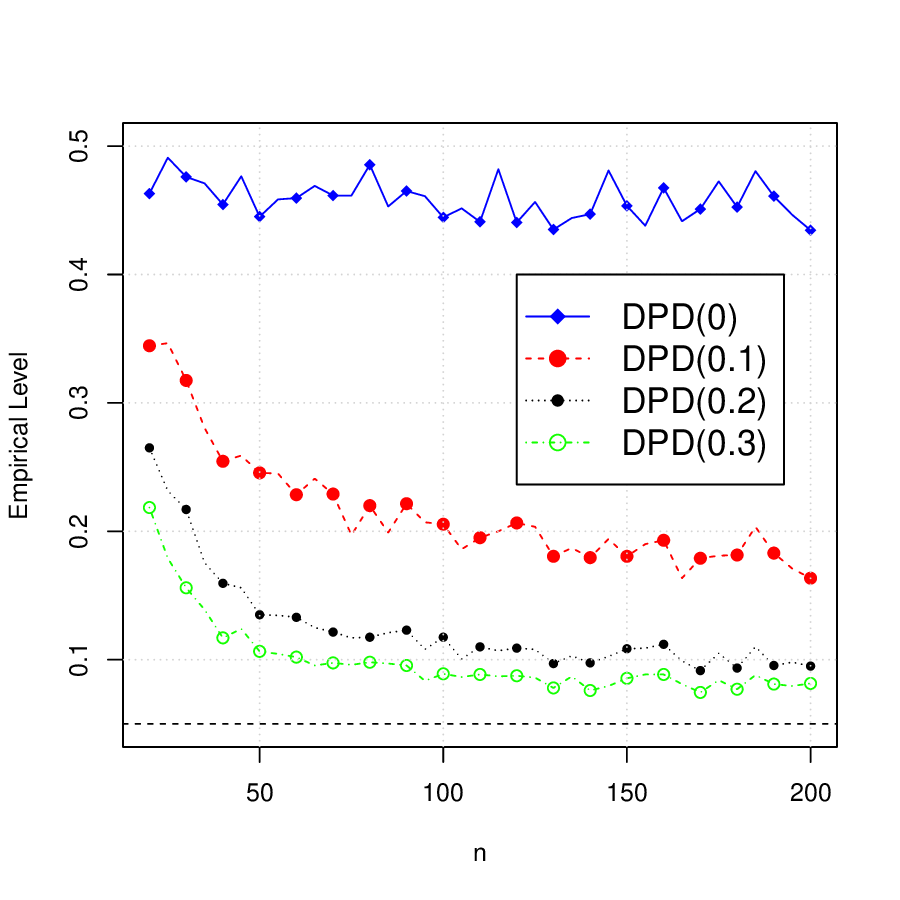}  \\
\vspace{-.9cm} (c) & (d)\\
 \includegraphics[height=6cm, width=7cm]{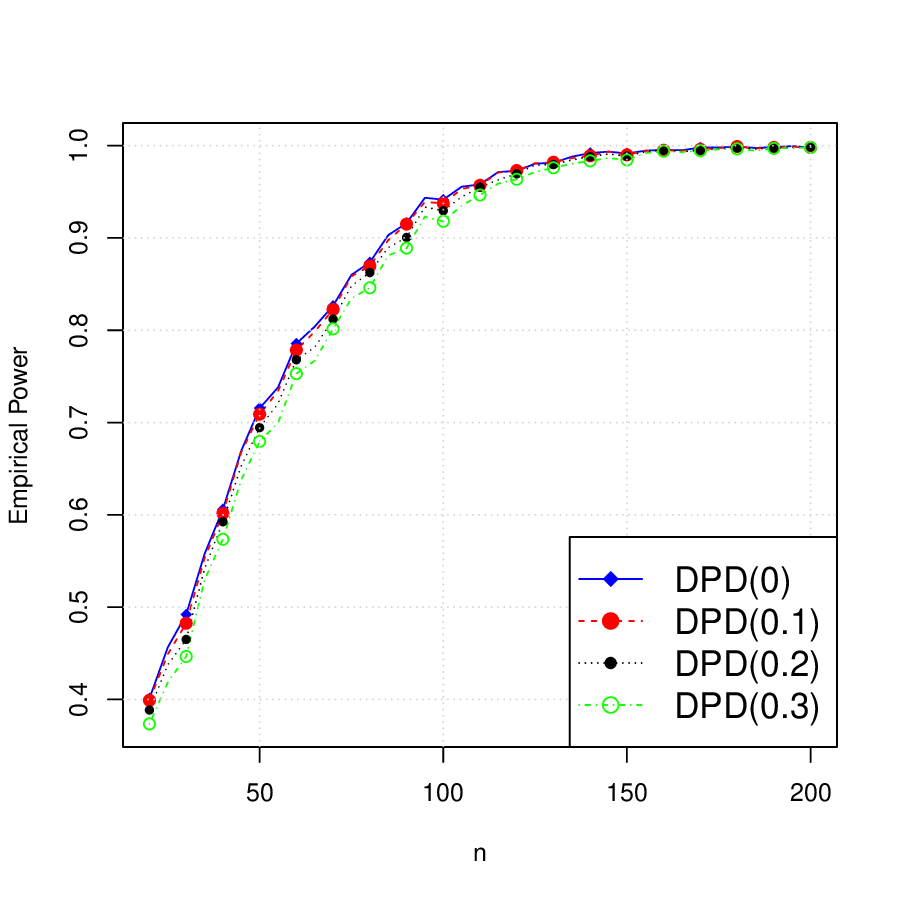} &
 \includegraphics[height=6cm, width=7cm]{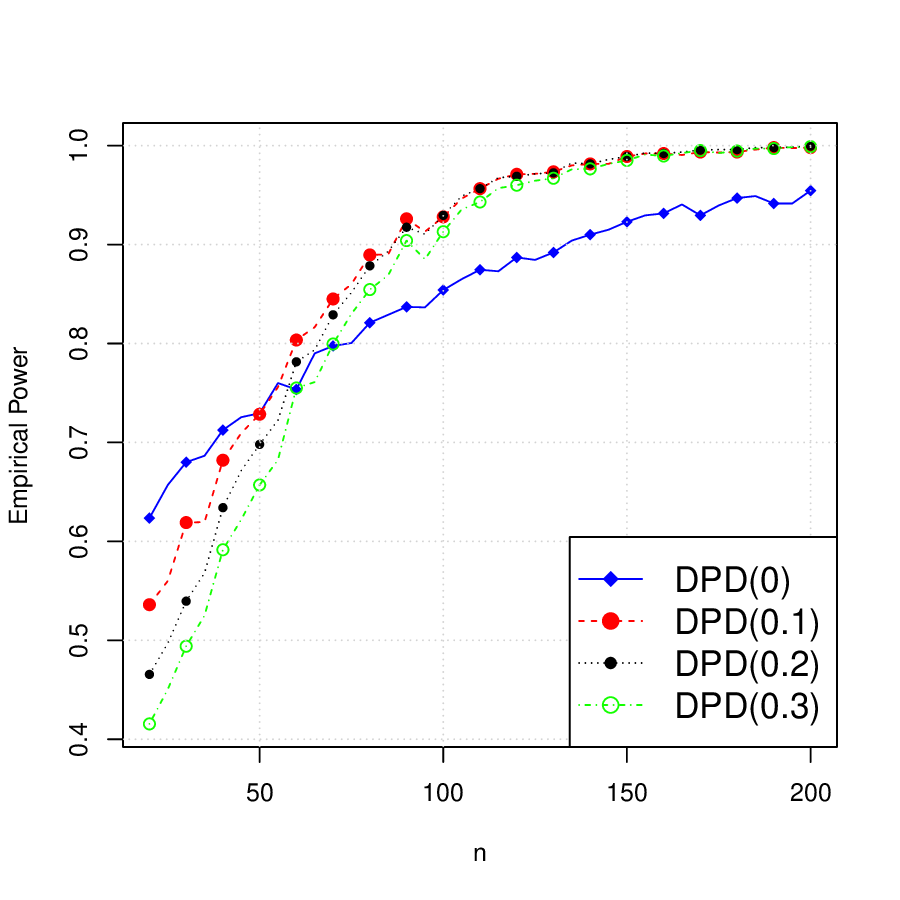}
\end{tabular}%
\caption{(a) Simulated levels of different tests for pure data; (b)
simulated levels of different tests for contaminated data; (c) simulated
powers of different tests for pure data; (d) simulated powers of different
tests for contaminated data.\label{fig:simulation}}
\end{figure}%

\begin{table}[h]
	\caption{The levels and powers of different Wald-type tests for sample size $n=100$, where $\epsilon$ is the proportion of contamination in the data. For power, the true parameter  $\boldsymbol{\beta}^* = \boldsymbol{\beta}_0 - c \boldsymbol{1}_{11}$, where $c = - 0.05$.}%
	\label{TAB:sim}
	\centering
	\begin{tabular}
		[c]{|c|c|cccc|}\hline
		Level or & $\epsilon$  &
		\multicolumn{4}{c|}{$\alpha$}\\
		Power &  & 0  & 0.1 & 0.2 & 0.3 \\ \hline
		Level & 0 & 0.056 &	0.049 &	0.044 &	0.043\\
		Power & 0 & 0.509 & 0.501 & 0.488 & 0.468\\
		Level & 0.05 & 0.989 &	0.077 &	0.064 &	0.070\\
		Power & 0.05 & 0.994 &	0.538 &	0.522 &	0.496\\
		\hline
	\end{tabular}
\end{table}

\begin{table}[h]
	\caption{The powers of the Wald-type tests for different sample sizes. The true parameter  $\boldsymbol{\beta}^* = \boldsymbol{\beta}_0 - c \boldsymbol{1}_{11}$, where $c = - 0.02$. }%
	\label{TAB:sim3}
	\centering
	\begin{tabular}
		[c]{|c|cccc|}\hline
		&\multicolumn{4}{c|}{$\alpha$}\\
		$n$  & 0  & 0.1 & 0.2 & 0.3 \\ \hline
		100 & 0.084 & 0.063 & 0.064 & 0.071\\
		200 & 0.165 & 0.120 & 0.127 & 0.136\\
		500 & 0.380 & 0.353 & 0.341 & 0.323\\
		1000 & 0.704 & 0.808 & 0.786 & 0.752\\
		1500 & 0.906 & 0.967 & 0.959 & 0.959\\
		2000 & 0.979 & 0.994 & 0.991 & 0.988\\
		3000 & 1.000 & 1.000 & 1.000 & 1.000
				
		\\\hline
	\end{tabular}
\end{table}

In the next set of simulation studies, we consider a more general set up to explore the performance of the proposed Wald-type tests. 
Here we have taken $k=10$;  the explanatory variables are generated independently from the standard normal distribution. 
To make the hypothesis general, we have arbitrarily chosen $k+1$ elements of vector $\boldsymbol{\beta}_0$,  
and $(k+1)\times k$ dimensional matrix $\boldsymbol{L}$. 
Each element of $\boldsymbol{\beta}_0$ and $\boldsymbol{L}$  is generated from an independent and identically distributed uniform distribution from $-1$ to 1. 
After that  $\boldsymbol{\beta}_0$ and $\boldsymbol{L}$ are kept unchanged throughout the simulation. 
Suppose $\boldsymbol{L}^{T} \boldsymbol{\beta}_0 = \boldsymbol{l}_0$. In the first simulation, 
$Y$ is generated from the Poisson distribution with mean parameter $\exp(\boldsymbol{X}^T \boldsymbol{\beta}_0)$;  
we are interested in verifying the levels of the Wald-type tests for testing the null hypothesis 
$H_0: \boldsymbol{L}^{T} \boldsymbol{\beta} = \boldsymbol{l}_0$. We have taken a sample of size $n=100$ and replicate it a 1000 times. 
The first row of Table \ref{TAB:sim} shows that the empirical levels of all four tests are closely bunched around the nominal level of $\alpha_0=0.05$. 
Next, we explore the powers of these tests when the true value of the parameter is in slight deviation from $\boldsymbol{\beta}_0$. 
We generated $Y$ from $\rm{Poisson}(exp(\boldsymbol{X}^T \boldsymbol{\beta}^*))$, 
where $\boldsymbol{\beta}^* = \boldsymbol{\beta}_0 - c \boldsymbol{1}_{11}$ with $c = - 0.05$. 
The results in the second row of Table \ref{TAB:sim} shows that the classical Wald test is the most powerful; 
however, other Wald-type tests  also produce very competitive powers. 
In Table \ref{TAB:sim3}, we expand the exploration of the study of power for pure data 
(as in the second row of Table \ref{TAB:sim}) over different sample sizes;  
the true parameter is taken very close to the null hypothesis where $\boldsymbol{\beta}^* = \boldsymbol{\beta}_0 - c \boldsymbol{1}_{11}$ with $c = - 0.02$. 
The result shows that the powers of all tests converge to one as sample size increases
indicating the consistency of the proposed tests.

To check the robustness properties of these tests, we contaminated $\epsilon$ proportion outliers in the $Y$ variable. 
Those outlying values are  25 standard deviations away from their respective means. 
The third row of Table \ref{TAB:sim} presents the empirical levels of the tests where there are 5\% outliers 
and for the rest of the data set  $Y \sim \rm{Poisson}(exp(\boldsymbol{X}^T \boldsymbol{\beta}_0))$. 
The classical Wald test shows an extreme inflation of level in this case, whereas other Wald-type tests show a stable level. 
In the same set up, we checked the powers of the tests under contamination where 95\% data are generated from  
$Y \sim \rm{Poisson}(exp(\boldsymbol{X}^T \boldsymbol{\beta}^*))$. 
The powers of the robust Wald-type tests are very similar to the corresponding uncontaminated case. 
So, it shows that 5\% contamination does not significantly affect the powers of these tests. 
Although, the observed power of the classical Wald test is very high, it is merely because of its inflated level. 
In fact, we could check that the actual level-corrected power is very poor in this situation.


While we have primarily used the influence function for the description of the robustness of our proposed tests, 
there are several other possible measures of robustness of statistical procedures. 
The breakdown point, which quantifies the degree of contamination that the procedure can withstand before it becomes completely uninformative, is one of them. 
Here we empirically explore the breakdown properties of our tests. 
In Table \ref{TAB:sim2}, the level robustness of the Wald type tests are demonstrated. 
The contamination scheme is as in the third row of Table \ref{TAB:sim}, but the contamination proportion is slowly allowed to increase to 0.5. 
Clearly the observed level for the ordinary Wald test is pushed to the maximum possible value at fairly small levels of contamination, 
but for moderately large values of $\alpha$ the observed levels remain substantially smaller than 1 even at $\epsilon = 0.5$.

\begin{table}[h]
	\caption{The levels of different Wald-type tests  for sample size $n=100$, where $\epsilon$ is the proportion of contamination in the data.}%
	\label{TAB:sim2}
	\centering
	\begin{tabular}
		[c]{|c|cccc|}\hline
		&\multicolumn{4}{c|}{$\alpha$}\\
		$\epsilon$  & 0  & 0.1 & 0.2 & 0.3 \\ \hline
		0 & 0.056 &	0.049 &	0.044 &	0.043\\
		0.05 & 0.989 &	0.077 &	0.064 &	0.070\\
		0.10 & 1.00 &	0.132 &	0.095 &	0.118\\
		0.15 & 1.00  & 0.222 & 0.163 & 0.180\\
		0.20 & 1.00 &	0.320 &	0.201 &	0.201\\
		0.25 & 1.00 & 0.445 & 0.288 & 0.287\\
		0.30 & 1.00 &	0.603 &	0.365 &	0.373\\
		0.35 & 1.00 & 0.731 & 0.461 & 0.475\\
		0.40 & 1.00 &	0.863 &	0.548 &	0.566\\
		0.45 & 1.00 & 0.928 & 0.627 & 0.647\\
		0.50 & 1.00 &	0.974 &	0.748 &	0.757		
		\\\hline
	\end{tabular}
\end{table}

Finally, we did a study on the effect of leverage points on the Wald-type tests. 
In the previous simulation studies, the explanatory variables are generated independently from the standard normal distribution. 
Now, $\epsilon$ proportion of explanatory variables in the samples (of size $n=100$) are generated independently from $N(\mu, 0.0001)$. 
The remaining set up of the simulation is same as the set up in the first row of Table \ref{TAB:sim}. 
Table \ref{TAB:sim4} shows the levels of the Wald-type tests for different values of $\epsilon$ and $\mu$. 
All simulated levels are very close to the nominal level of $\alpha_0=0.05$, 
so the result demonstrates that at least in this study these tests are robust against  leverage points. 

\begin{table}[h]
	\caption{The levels of the Wald-type tests for different contaminated proportions ($\epsilon$) and mean shift ($\mu$) for the leverage points.  The sample size is $n=100$.}%
	\label{TAB:sim4}
	\centering
	\begin{tabular}
		[c]{|c|c|cccc|}\hline
		&&\multicolumn{4}{c|}{$\alpha$}\\
		$\epsilon$ & $\mu$  & 0  & 0.1 & 0.2 & 0.3 \\ \hline
		0 & 0 & 0.056 &	0.049 &	0.044 &	0.043\\
		0.05 & 3 & 0.049 & 0.027 & 0.034 & 0.044\\
		0.05 & 4 & 0.054 & 0.041 & 0.041 & 0.040\\
		0.10 & 3 & 0.049 & 0.039 & 0.042 & 0.045\\
		0.10 & 4 & 0.044 & 0.040 & 0.043 & 0.050		
		\\\hline
	\end{tabular}
\end{table}

%

\section{Real Data Examples}

\subsection{Credit Cards Data}

As the first application of our proposed method, we consider a benchmark dataset from Agresti (2018),
which consists of a random sample from an Italian study conducted to investigate 
the relation of holding a travel credit card (such as Diners Club or American Express) with individual's personal income.
The data are given for 31 possible values of annual income (in millions of lira, the previous currency of Italy),
where the number of total persons sampled and the number of them having at least one card are recorded at each income level. 
These data have been traditionally analyzed through either logistic or Poisson regression models.

\begin{table}[!b]
	\caption{The MDPDEs of $\boldsymbol{\beta}$, their standard errors (in parenthesis) and the p-values of their significance testing
		obtained by the proposed MDPDE-based Wald-type tests.}%
	\label{TAB:card}\vspace{5pt}
	\centering
	\resizebox{\textwidth}{!}{
		\begin{tabular}{|l|ccccc|ccccc|}\hline
			& \multicolumn{5}{c}{MDPDE (standard error)}	& \multicolumn{5}{|c|}{p-values for significance testing}\\ \hline	
			$\alpha$ 			& 0  & 0.1 & 0.3 & 0.5 & 0.7 &  0 & 0.1 & 0.3 & 0.5 & 0.7\\ \hline
			\multicolumn{11}{|l|}{\textbf{Pure Data}}\\
			Intercept  			& $-$2.737	&	$-$2.274	&	$-$2.039	&	$-$2.019	&	$-$2.016	&	0.00001	&	0.00005	&	0.00032	&	0.00094	&	0.00186	\\
			($\beta_0$)			& (0.56)	&	(0.56)	&	(0.57)	&	(0.61)	&	(0.65)	&		&		&		&		&		\\
			Income 				& 0.021	&	0.018	&	0.017	&	0.015	&	0.015	&	0.00004	&	0.00045	&	0.00185	&	0.01170	&	0.02257	\\	
			($\beta_1$)			& (0.01)	&	(0.01)	&	(0.01)	&	(0.01)	&	(0.01)	&		&		&		&		&		\\
			\textit{LOG-CASE} 	& 1.215	&	1.051	&	0.940	&	1.028	&	0.999	&	0.00000	&	0.00002	&	0.00013	&	0.00010	&	0.00035	\\
			($\beta_2$)			& (0.24)	&	(0.24)	&	(0.25)	&	(0.26)	&	(0.28)	&		&		&		&		&		\\\hline
			\multicolumn{11}{|l|}{\textbf{With One Outlier}}\\
			Intercept 			& $-$0.708	&	$-$2.069	&	$-$2.040	&	$-$2.009	&	$-$2.022	&	0.10434	&	0.00014	&	0.00036	&	0.00101	&	0.00197	\\
			($\beta_0$)			& (0.44)	&	(0.54)	&	(0.57)	&	(0.61)	&	(0.65)	&		&		&		&		&		\\
			Income 				& 0.009	&	0.018	&	0.017	&	0.016	&	0.015	&	0.10846	&	0.00091	&	0.00301	&	0.00776	&	0.02832	\\
			($\beta_1$)			& (0.01)	&	(0.01)	&	(0.01)	&	(0.01)	&	(0.01)	&		&		&		&		&		\\
			\textit{LOG-CASE} 	& 0.516	&	0.954	&	0.977	&	0.920	&	1.011	&	0.00646	&	0.00005	&	0.00008	&	0.00048	&	0.00034	\\
			($\beta_2$)			& (0.19)	&	(0.24)	&	(0.25)	&	(0.26)	&	(0.28)	&		&		&		&		&		\\
			\hline
		\end{tabular}
	}
\end{table}

It has been justified that the number of people having at least one travel card ($Y$) can be modeled well through 
a Poisson regression model with significant covariates being their income and 
logarithm (\textit{LOG-CASE}) of the total number of people sampled at the same income level (and intercept).
We have also used the same model with $Y$ having a Poisson distribution with its mean being given by the regression structure
$$
\log E[Y] = \beta_0 + \beta_1 (\mbox{Income}) + \beta_2 (\mbox{\textit{LOG-CASE}}). 
$$
We have estimated these regression coefficients $\boldsymbol{\beta}=(\beta_0, \beta_1, \beta_2)^T$
by our MDPDE at different values of $\alpha>0$, which are presented in Table \ref{TAB:card} along with their standard errors (SEs) and 
the p-values  for testing the significance of individual regression coefficients (i.e., $H_0 : \beta_i = 0$)
obtained through our proposed Wald-type test.
The data do not appear to have any major natural outliers. 
So, to illustrate the claimed robustness of our proposal, we have changed one response value (at the lowest income level) 
from 0 to 10 and repeated the estimation and testing exercise for these contaminated data which are also presented in Table \ref{TAB:card}. 
Note that the column $\alpha=0$ refers to the MLE and the p-values obtained by the usual Wald test.

We can observe from Table \ref{TAB:card} that the MDPDEs are very close to the usual MLE ($\alpha=0$) for the pure data without any contamination
but their standard error increases slightly with increasing values of $\alpha>0$ as expected from our theoretical discussions.
Further, the proposed MDPDE based Wald-type tests at any $\alpha>0$ also yield p-values close to those from usual Wald tests ($\alpha=0$)
indicating the significance of all three regression coefficients $(\beta_0, \beta_1, \beta_2)$ at any reasonable level.
However, with the introduction of just one outlier in a data set of 31 observations ($< 3.5\%$ contamination), 
the MLEs of all the regression coefficients change drastically 
whereas the MDPDEs with $\alpha>0$ are only minimally altered indicating their robust nature. 
Similarly, this small amount of contamination also drastically changes the p-values obtained from the usual Wald test
which now fails to indicate the significance of $\beta_0$ and $\beta_1$ even at the 10\% level.
In contrast, our proposed MDPDE based Wald-type tests  provide much stable p-values for all positive values of $\alpha$ 
and successfully indicate the (true) significance of all regression coefficients even under contamination
justifying their claimed robustness advantages.

\subsection{Epilepsy Data}

Our next illustration is another popular clinical trial data which itself contains few outlying observations
(Leppik et al., 1985; Thall and Vail, 1990). 
We model the total number of epilepsy attacks of 59 patients by a Poisson regression model 
with the available covariates, which are the treatment indicator (versus the control group), 
the eight-week baseline seizure rate (in multiple of 4) prior to randomization, 
the age of the patient (in multiple of 10 years)
and the interaction of treatment with the baseline seizure rate.
These data have been studied by several researchers dealing with robust inference in the Poisson model 
(e.g., Cantoni and Ronchetti, 2001; Hosseinian, 2009; Ghosh and Basu, 2016).
Unlike the credit cards data which does not have any natural outliers, 
here it is observed that there are some outlying observations in the data 
which cause the interaction effects to be insignificant and the coefficient of age to be significant
in classical maximum likelihood based inference,
but any robust methodology yields the opposite inference.

We apply our proposed MDPDE based Wald-type tests for testing the significance of the coefficients of age 
and the interaction effects to see if the correct inference can be obtained even in the presence of outliers.
The resulting p-values are plotted over $\alpha\geq 0$ in Figure \ref{FIG:Edata}.
Clearly, the classical Wald test (at $\alpha=0$) provides incorrect inference at the 5\% level  
in indicating the significance of the age effect and insignificance of the interaction effect between the treatment and the baseline seizure rate.
However, our proposed Wald-type tests with positive values of $\alpha$, approximately in the range 0.3 and above, 
again provide robust (correct) inference under data contamination
indicating a significant effect of the interaction between the treatment and  the baseline seizure rate 
on the number of epilepsy attack of a patient along with insignificant effect of patient's age.
This again illustrates the applicability and advantages of our proposal in getting stable and correct insights 
from any real-life dataset even in the presence of possible outliers in the data.

\begin{figure}[h]
	\centering
	\includegraphics[width=0.6\textwidth]{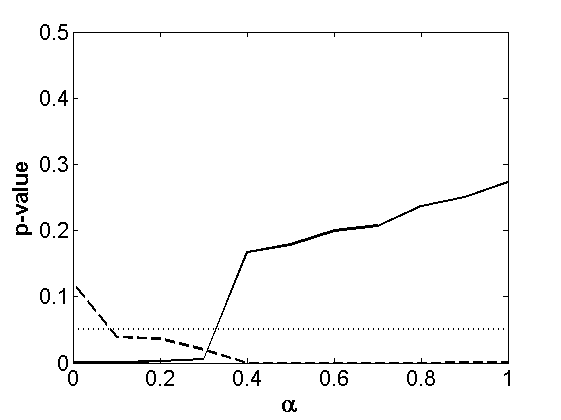}
	\caption{P-values obtained by the MDPDE based Wald-type tests at different $\alpha\geq 0$ 
		for testing the significance of the effect of patient's age (solid line) and the interaction effect 
		between the treatment and baseline seizure rate (dashed line) for the epilepsy data. 
		The dotted line represents the 5\% level.}%
	\label{FIG:Edata}%
\end{figure}

\section{Conclusion and discussions}

The class of generalized linear models represents a very important component of the statistical methodology toolbox. 
In this paper we have dealt with robust tests for testing any general composite null hypothesis 
in the generalized linear models  under the stochastic covariate set up. 
For this purpose, the family of density power divergences have been utilized; 
this results in a collection of Wald-type tests which includes the classical Wald test as a special case, 
but also accommodates other, more robust solutions, some of which attain a very high degree of robustness 
with little loss in power  relative to the classical Wald test for the pure data scenario. 
The asymptotic properties of these tests and their theoretical robustness have been rigorously established. 
We have chosen the Poisson regression model for analyzing count data as the medium of demonstration; 
numerical results illustrating the performances in terms of level and power under different scenarios 
and graphical results illustrating the nature and behavior of the influence functions 
clearly establish the usefulness of our proposed tests.

It is important to note that the proposed test directly depends on the MDPDE
and so some comments about its computation is needed here.
Clearly, the loss function of MDPDEs may have several local minima and 
hence the corresponding estimating equation may have more than one solution. 
So, in order to obtain the global minimizer as the MDPDE for general data applications, 
it is necessary to try different starting values of the optimization algorithm 
and choose the solution having minimum value of the DPD loss function;
these often help  to find  the absolute minimum with a certain probability depending 
on the number and structure of the starting parameter values used. 
This is one advantage of the MDPDE over general M-estimators defined only in terms of estimating equations,
since there may not be a easy way to choose from the multiple roots of those estimating equations. 
However, there is still the requirement of more research and discussion on the computation  
of the MDPDE as well as in terms of obtaining an efficient algorithms for the same purpose,
since the choice of starting values is not clear and may be time consuming. 
We hope to consider such computational aspects further in our future work.

As we have mentioned briefly in Section \ref{sec2.2}, 
our present work examines the robustness of the proposed estimators and tests of hypotheses theoretically 
in terms of boundedness of influence function, which indeed only guarantees their local B-robustness.
We have provided empirical illustrations for the influence function
and the contamination bias for finite sample illustrations. 
However, there are several other robustness measures defined from different perspective,
including breakdown point, V-robustness etc., which are as crucial in examining the robustness properties.
We have provided some limited illustrations of the breakdown property in our numerical illustrations. 
It would, however, be an interesting future work to verify these measures (including breakdown) theoretically for our MDPDE
and the associated Wald-type tests. 
This would also represent an interesting future work. 

Finally, we emphasize again that this work investigated the robustness of the proposed MDPDE and Wald-type tests
against data contamination (e.g., outliers). It would be important to investigate the robustness of these procedures
in other aspects as well, e.g., against misspecification of the model or the design matrix or 
any other assumptions including the linearity of the covariates within the GLM. 
It can be intuitively said that  wrongly specifying the design matrix to be fixed 
while it is random would have the similar effects on the MDPDE as well as on the MLE described in the introduction. 
On the other hand, since these present MDPDE based methods are developed with particular focus on data contamination, 
other non-parametric procedures might outperform them in case of a complete misspecification of the underlying model.
However, more research is surely needed to examine the extent of model misspecification that our MDPDE can tolerate
which we hope to consider in a sequel paper. 

\bigskip\bigskip
\noindent
\textbf{Acknowledgments:}\\
We would like to thank two anonymous referees for their helpful comments and suggestions which have improved the paper. 
This research has been partially supported by Grant PGC2018-095194-B-100 from Ministerio de Ciencia, Innovacion y Universidades (Spanish government).
The work of AG is also partially supported by the INSPIRE Faculty research grant from Department of Science and Technology, Government of India.

\end{document}